\documentclass[11pt]{article}
\usepackage{tikz}
\usetikzlibrary{shapes,arrows,shadows,positioning,matrix,patterns,decorations.markings}
\usepackage{amssymb}
\usepackage{enumerate}
\usepackage{comment}
\usepackage{fancyhdr}
\newtheorem{theorem}{Theorem}[section]
\newtheorem{lemma}[theorem]{Lemma}

\newtheorem{observation}[theorem]{Observation}
\newtheorem{definition}{Definition}

\newtheorem{step}{Step}

\newenvironment{proof}{\noindent\textbf{Proof.}}{{}\hfill$\Box$\\}
\newenvironment{proofof}{\noindent\textbf{Proof of}}{{}\hfill$\Box$\\}

\newcommand{\anchor}{\texttt{anchor}}

\newcommand{\cE}{\mathcal{E}}
\def\oY4{\protect\overrightarrow{Y_4^\infty}}

\def\Y4{Y_4^\infty}

\newcommand{\edge}[1]{(\protect\overrightarrow{#1})}
\newcommand{\cpath}[2]{\texttt{cpath}_{#1}(#2)}
\newcommand{\nodeList}{{(396,87)/0/1/1/1/0/1/1/1/0/right},{(309,180)/1/1/1/2/1/1/1/1/1/right},{(316,353)/2/0/2/1/1/0/1/1/1/right},{(301,441)/3/0/1/1/1/0/1/1/1/left},{(127,248)/4/1/1/1/1/1/1/1/1/right},{(323,305)/5/1/1/1/2/1/1/1/1/right},{(423,410)/6/0/1/1/0/0/1/1/0/right},{(204,408)/7/1/0/1/1/1/0/1/1/left},{(209,80)/8/1/1/0/1/1/1/0/1/right},{(50,459)/9/0/0/0/1/0/0/0/1/right},{(382,366)/10/1/1/1/1/1/1/1/1/right},{(48,160)/11/1/0/1/1/1/0/1/1/left},{(377,433)/12/0/1/2/1/0/1/1/1/right},{(332,72)/13/1/0/1/0/1/0/1/0/left},{(217,175)/14/1/1/1/2/1/1/1/0/right},{(250,350)/15/1/1/1/1/1/1/1/1/right},{(289,253)/16/1/1/1/1/0/1/1/1/right},{(147,312)/17/1/1/0/1/1/1/0/1/left},{(273,211)/18/1/1/1/1/1/1/1/1/right},{(375,208)/19/2/1/1/0/1/1/1/0/right},{(71,270)/20/1/1/0/1/1/1/0/1/left},{(392,328)/21/0/1/0/2/0/1/0/1/right},{(185,305)/22/1/1/1/0/1/1/1/0/right},{(247,198)/23/1/0/1/1/1/0/1/0/left},{(276,276)/24/1/1/1/1/1/1/1/1/left},{(190,132)/25/1/1/0/1/1/1/0/1/right},{(105,154)/26/0/1/1/1/0/1/1/1/right},{(30,100)/27/1/0/0/1/1/0/0/1/right},{(84,61)/28/1/1/0/0/1/1/0/0/right}}
\newcommand{\YaoEdges}{0/13/red!70,1/14/red!70,2/15/red!70,3/15/red!70,4/11/red!70,5/24/red!70,6/10/red!70,7/17/red!70,8/28/red!70,9/11/red!70,10/5/red!70,11/27/red!70,12/2/red!70,13/28/red!70,14/25/red!70,15/22/red!70,16/18/red!70,17/4/red!70,18/23/red!70,19/1/red!70,20/11/red!70,21/5/red!70,22/4/red!70,23/14/red!70,24/18/red!70,25/28/red!70,26/27/red!70,1/0/blue,2/5/blue,3/12/blue,4/14/blue,5/19/blue,7/15/blue,8/13/blue,9/17/blue,10/21/blue,11/26/blue,12/6/blue,14/13/blue,15/5/blue,16/1/blue,17/22/blue,18/1/blue,19/0/blue,20/4/blue,21/0/blue,22/24/blue,23/1/blue,24/16/blue,25/8/blue,26/25/blue,27/28/blue,0/6/red!70,1/19/red!70,2/10/red!70,4/22/red!70,5/10/red!70,7/3/red!70,8/14/red!70,10/6/red!70,11/4/red!70,13/0/red!70,14/23/red!70,15/2/red!70,16/5/red!70,17/7/red!70,18/16/red!70,19/21/red!70,20/17/red!70,21/6/red!70,22/15/red!70,23/18/red!70,24/5/red!70,25/14/red!70,26/4/red!70,27/11/red!70,28/26/red!70,0/1/blue,1/18/blue,2/3/blue,3/9/blue,4/20/blue,5/2/blue,6/12/blue,7/9/blue,8/25/blue,10/12/blue,12/3/blue,13/1/blue,14/4/blue,15/7/blue,16/24/blue,17/9/blue,18/22/blue,19/16/blue,20/9/blue,21/10/blue,22/17/blue,23/22/blue,24/15/blue,25/26/blue,26/11/blue,28/27/blue}
\newcommand{\stepTwoEdges} {21/0,23/1,12/2,2/3,7/3,16/5,10/12,8/13,1/14,8/14,20/17,24/18,5/19,19/21}
\newcommand{\stepThreeEdges} {21/0,5/19,2/3,10/12,7/3,8/13,20/17,24/18}
\newcommand{\shortcuts} {1/8}
\newcommand{\pseudoStrongEdges} {28/26/red!70,19/16/blue,24/15/blue}
\newcommand{\anchorEdges} {0/13/red!70,3/15/red!70,4/11/red!70,5/24/red!70,6/10/red!70,7/17/red!70,10/5/red!70,14/25/red!70,15/22/red!70,17/4/red!70,18/23/red!70,19/1/red!70,21/5/red!70,26/27/red!70,1/0/blue,3/12/blue,4/14/blue,7/15/blue,9/20/blue,10/21/blue,11/26/blue,12/6/blue,15/5/blue,16/1/blue,17/22/blue,20/4/blue,22/24/blue,24/16/blue,25/8/blue,26/25/blue,27/28/blue,1/19/red!70,2/10/red!70,4/22/red!70,5/10/red!70,11/4/red!70,13/0/red!70,14/23/red!70,15/2/red!70,17/7/red!70,18/16/red!70,22/15/red!70,23/18/red!70,27/11/red!70,28/26/red!70,0/1/blue,1/18/blue,4/20/blue,5/2/blue,8/25/blue,13/1/blue,14/4/blue,15/7/blue,16/24/blue,19/16/blue,21/10/blue,22/17/blue,24/15/blue,25/26/blue,26/11/blue,28/27/blue}
\newcommand{\StrongEdges} {0/13/red!70,4/11/red!70,5/24/red!70,6/10/red!70,7/17/red!70,10/5/red!70,14/25/red!70,15/22/red!70,18/23/red!70,19/1/red!70,1/0/blue,3/12/blue,4/14/blue,7/15/blue,9/20/blue,10/21/blue,11/26/blue,12/6/blue,17/22/blue,20/4/blue,24/16/blue,25/8/blue,26/25/blue,27/28/blue,1/19/red!70,4/22/red!70,5/10/red!70,11/4/red!70,13/0/red!70,14/23/red!70,15/2/red!70,17/7/red!70,18/16/red!70,22/15/red!70,23/18/red!70,27/11/red!70,0/1/blue,1/18/blue,4/20/blue,5/2/blue,8/25/blue,14/4/blue,15/7/blue,16/24/blue,21/10/blue,22/17/blue,25/26/blue,26/11/blue,28/27/blue}

\newcommand{\charge}[4]{\arraycolsep=1.2pt\def\arraystretch{0.75}\begin{array}{cc}
  #2 & #1 \\
  #3 & #4
 \end{array}}

\usetikzlibrary{arrows,shapes,positioning}
\usetikzlibrary{decorations.markings}
\tikzstyle arrowstyle=[scale=1]
  \tikzstyle{vertex}=[fill=black, circle,inner sep = 1.5pt,]
  \tikzstyle{selected vertex} = [vertex, fill=red!24]
  \tikzstyle{selected edge} = [draw,line width=5pt,-,red!50]
  \tikzstyle{edge} = [draw,ultra thick,-,black]
\tikzstyle anchorEdge=[ultra thick, postaction={decorate,decoration={markings,
    mark=at position 1 with {\arrow[arrowstyle, scale=1.5]{stealth}}}}]
\tikzstyle directed=[postaction={decorate,decoration={markings,
    mark=at position .5 with {\arrow[arrowstyle, scale=1.3]{stealth}}}}]
\tikzstyle strongEdge=[anchorEdge]
\tikzstyle pseudoStrongEdge=[anchorEdge]
\tikzstyle shortcutEdge=[edge]
\tikzstyle addedEdge=[anchorEdge, color=purple, dotted]
 \tikzstyle{yaoEdge} = [edge, thick, directed]

\title{There are Plane Spanners of Maximum Degree 4}

\author{
Nicolas Bonichon\thanks{LaBRI, Universit\'{e} Bordeaux
  1. {bonichon@labri.fr}. This work was partially supported by ANR grant JCJC EGOS ANR-12-JS02-002-01.} \hspace{0.75cm}
Iyad Kanj\thanks{School of Computing, DePaul University. {\{ikanj,lperkovic\}@cs.depaul.edu}.} \hspace{0.75cm}
  Ljubomir Perkovi\'{c}\footnotemark[2] \hspace{0.75cm}
  Ge Xia\thanks{Department of Computer Science, Lafayette College. {xiag@lafayette.edu}.}
}

\date{}

\begin{document}

\maketitle
 
\thispagestyle{empty}

\begin{abstract}
Let $\cE$ be the complete Euclidean graph on a set of points embedded
in the plane. Given a constant $t \geq 1$, a spanning subgraph $G$
of $\cE$ is said to be a $t$-\emph{spanner}, or simply a spanner,
if for any pair of vertices $u,v$ in $\cE$ the distance between $u$ and
$v$ in $G$ is at most $t$ times their distance in $\cE$. A spanner is
{\em plane} if its edges do not cross.

This paper considers the question: ``What is the smallest {\em maximum degree}
that can always be achieved for a {\em plane} spanner of $\cE$?''
Without the planarity constraint, it is known that the answer is 3 which is
thus the best known lower bound on the degree of any plane spanner. With the
planarity requirement, the best known upper bound on the maximum degree is 6,
the last in a long sequence of results improving the upper bound. In this paper
we show that the complete Euclidean
graph always contains a plane spanner of maximum degree at most 4 and make a big
step toward closing the question. Our construction leads to an efficient
algorithm for obtaining the spanner from Chew's $L_1$-Delaunay triangulation.
\end{abstract}

\newpage
\setcounter{page}{1}
\section{Introduction}
\label{sec:intro}

Let $\cE$ be the complete Euclidean graph on a set of points $P$ embedded
in the plane. Given a constant $t \geq 1$, a spanning subgraph $G$ of $\cE$ is
said to be a $t$-\emph{spanner}, or simply a spanner, if for any
pair of vertices $u,v$ in $\cE$ the distance between $u$ and $v$ in $G$
is at most $t$ times their distance in $\cE$. The constant $t$ is referred
to as the {\em stretch factor}. A spanner is {\em plane}
if its edges do not cross.

In this paper, we consider the following question: {\em What is the smallest
maximum degree that can always be achieved for plane spanners of complete
Euclidean graphs?} Or, to put it more precisely: {\em What is the smallest $d$
such that for some constant $t \geq 1$ there always exists a plane
$t$-spanner of maximum degree at most $d$ on any set of points on the plane?}
This fundamental question was raised by Bose and Smid~\cite{BS10} in their
recent survey of geometric problems. It is a natural extension to classical
questions on spanners of complete Euclidean graphs, and Delaunay triangulations
in particular.

In the mid-1980s, the fundamental question of  whether a plane spanner of $\cE$
always exists was considered. In his seminal 1986 paper, Chew answered the question in the
affirmative~\cite{Che86}. He proved, in particular, that the $L_1$-Delaunay
triangulation of $P$, i.e. the dual of the Voronoi diagram of $P$ based on
the $L_1$-distance, is a $\sqrt{10}$-spanner of $\cE$.
Chew's result was followed by a series of papers demonstrating that other
Delaunay triangulations are plane spanners as well. In 1987, Dobkin
{\it et al.}~\cite{DFS90} were successful in showing that the (classical)
$L_2$-Delaunay triangulation of $P$, i.e. the dual of the Voronoi diagram
of $P$ based on the $L_2$-distance (i.e., the Euclidean distance) is a spanner
as well. The bound on
the stretch factor they obtained was subsequently improved by Keil and
Gutwin~\cite{KG92} as shown in Table~\ref{ta:related}. In the meantime,
Chew~\cite{Che89} showed that the $TD$-Delaunay triangulation---again a dual
of a Voronoi diagram but this time defined using a distance function
based on an equilateral triangle rather than a square ($L_1$-distance) or a
circle ($L_2$-distance)---is a $2$-spanner.

The bound on the stretch factor of an $L_2$-Delaunay triangulation by Keil
and Gutwin stood unchallenged for many years until Xia recently improved the
bound to below 2~\cite{Xia13} (see Table~\ref{ta:related}).
Recently as well, Bonichon {\em et al.}~\cite{BGHP12} improved Chew's
original bound on the stretch factor of the $L_1$-Delaunay triangulation to
$\sqrt{4+2\sqrt{2}}$ and showed this bound to be tight.

\begin{table}[!b]
\label{ta:related}
\begin{center}
\begin{tabular}{lrr}
{\bf Paper} &  {\bf Spanner} & {\bf Stretch factor bound} \\ \hline
Chew~\cite{Che86} & $L_1$-Delaunay & $\sqrt{10} \approx 3.16$ \\ \hline
Bonichon {et al.}~\cite{BGHP12} & $L_1$-Delaunay &
 $\mathbf{\sqrt{4+2\sqrt{2}}\approx 2.61}$ \\ \hline
Dobkin {\it et al.}~\cite{DFS90} & $L_2$-Delaunay & $\frac{\pi(1+\sqrt{5})}{2} \approx 5.08$ \\ \hline
Keil \& Gutwin~\cite{KG92} & $L_2$-Delaunay & $\frac{4\pi}{3\sqrt{3}} \approx 2.42$ \\ \hline
\cite{Xia13} & $L_2$-Delaunay & $1.998$ \\ \hline
Chew~\cite{Che89} & $TD$-Delaunay & $\mathbf{2}$ \\ \hline
\end{tabular}
\end{center}
\caption{Key results on (unbounded degree) plane spanners; tight bounds are
in bold.}
\end{table}

Minimizing the stretch factor of a plane spanner of $\cE$ is one natural goal.
Another
one is minimizing the maximum degree of the plane spanner. This restriction
eliminates, for example, the various Delaunay triangulations because they can
have unbounded degree. The lower bound on the maximum degree of a spanner is 3,
because a Hamiltonian path through a set of $n$ points arranged in a grid has
stretch factor $\Omega(\sqrt{n})$. Work on bounded degree {\em but not
necessarily
plane} spanners of $\cE$ closely followed the above-mentioned work on plane
spanners. In a 1992 breakthrough, Salowe \cite{Sal94} proved the existence
of spanners of maximum degree at most 4. The question was then resolved by Das
and Heffernan~\cite{DH96} who showed that spanners of maximum degree at most 3
always exist.

The focus in this line of research was to prove the existence of low degree
spanners and the techniques developed to do so were not tuned towards
constructing spanners that had both low degree {\em and} low stretch factor.
Furthermore, the bounded-degree spanners shown to exist were not guaranteed to be plane.
In recent years, bounded degree plane spanners have been used as the
building block of wireless network topologies. Emerging wireless
distributed system technologies, such as wireless ad-hoc and
sensor networks, are often modeled as proximity graphs in the Euclidean
plane. Spanners of proximity graphs represent topologies that can be
used for efficient unicasting, multicasting, {\em and/or}
broadcasting. For these applications, in addition to low stretch factor,
spanners are typically required to be plane and have bounded degree.
The planarity requirement is for efficient routing (see~\cite{planerouting}),
while the bounded degree requirement is due to the physical limitations of
wireless devices (see~\cite{boundedrouting}).

\begin{table}
\label{ta:related2}
\begin{center}
\begin{tabular}{lrr}
{\bf Paper} & {\bf $\Delta$} & {\bf Stretch factor bound} \\ \hline
Bose {\it et al.}~\cite{BGS05a} & 27 & $(\pi+1) C_0 \approx 8.27$ \\ \hline
Li and Wang~\cite{LW04} & 23 & $(1+\pi \sin \frac{\pi}{4}) C_0 \approx 6.43$ \\ \hline
Bose {\it et al.}~\cite{BSX09} & 17 & $(2+2\sqrt{3} + \frac{3\pi}{2} +
2\pi\sin(\frac{\pi}{12})) C_0 \approx 23.56 $\\ \hline
Kanj and Perkovi\'{c}~\cite{KP08} & 14 &
$(1+\frac{2\pi}{14\cos(\frac{\pi}{14})}) C_0 \approx 2.91$\\ \hline
Bonichon {\it et al.}~\cite{BGHP10} & 6 & 6 \\ \hline
Bose {\it et al.}~\cite{BCC12} & 6 & $1/(1-\tan(\pi/7)(1+1/\cos(\pi/14)))C_0 \approx 81.66$ \\ \hline
{\bf This paper} & 4 & $\sqrt{4+2\sqrt{2}} (1+\sqrt{2})^2 (3+\sqrt 2)^6 \approx 112676$ \\ \hline
\end{tabular}
\end{center}
\caption{Results on plane spanners with maximum degree bounded by $\Delta$. The
constant $C_0 = 1.998$ is the best known upper
bound on the stretch factor of the $L_2$-Delaunay triangulation~\cite{Xia13}.
The stretch factor bound in this paper can be made much tighter with a more
careful analysis.}

\end{table}

Bose {\it et al.}~\cite{BGS05a} were the first to show how to extract
a spanning subgraph of the classical $L_2$-Delaunay triangulation that is a
bounded-degree, plane spanner of $\cE$. The maximum degree and stretch factor
bounds they obtained were subsequently improved by Li and Wang~\cite{LW04}, by Bose
{\it et al.}~\cite{BSX09}, and by Kanj and Perkovi\'{c}~\cite{KP08}
(see all bounds in Table~\ref{ta:related2}). The approach used in all
these results was to extract a bounded degree spanning subgraph of the classical
$L_2$-Delaunay triangulation and the main goal was to obtain a bounded-degree
plane spanner of $\cE$ with the smallest possible stretch factor.

Recently, Bonichon {\it et al.}~\cite{BGHP10} focused on lowering the bound
on the maximum degree of a plane spanner and developed a new approach.
Instead of using the classical $L_2$-Delaunay triangulation
as the starting point of the spanner construction, they used the $TD$-Delaunay
triangulation defined by Chew~\cite{Che89}. They achieved a significant
decrease in the bound on the maximum degree: from 14 down to 6. The plane
spanner they constructed also had a surprisingly small stretch factor of 6.
Independently, Bose {\it et al.}~\cite{BCC12} have also been able
to obtain a plane spanner of maximum degree at most 6, by starting from
the $L_2$-Delaunay triangulation; the spanner they obtain has the additional
property of being {\em strong} which means that between every pair of
vertices $u$ and $v$ there is, in the spanner, a path that consists of edges
whose length is no more than the Euclidean distance between $u$ and $v$.

In this paper, we push the bound on the maximum degree of a plane spanner
from 6 down to 4 and make a big step toward closing a fundamental question.
Interestingly, the starting point for our spanner construction is Chew's
original $L_1$-Delaunay triangulation, a graph that has been largely
overlooked in the last quarter century. We define this triangulation, and the
equivalent $L_\infty$-Delaunay triangulation, in the next section. In
Section~\ref{se:yao}, we introduce a key tool: a directed
version of the $L_\infty$-distance-based Yao graph $\Y4$ introduced by Bose
{\it et al.}~\cite{BDD10}. {\em En passant}, we prove that $\Y4$ is a plane
$\sqrt{20+14\sqrt{2}} \approx 6.3$-spanner of $\cE$. Then,
in Section~\ref{se:anchors}, we define
{\em standard paths} between the endpoints of every edge in $\Y4$. In
Section~\ref{se:H8}, we construct a subgraph $H_8$ of $\Y4$ of maximum
degree at most 8 and show that it is a
spanner by proving that it contains short standard paths.
Finally, in Section~\ref{se:H4}, we show that some edges in $H_8$ are
redundant and we remove them, while adding new shortcut edges, to obtain
$H_4$, a spanner of maximum degree at most 4. While the proofs in the paper
are quite
technical, our construction leads to a simple and efficient algorithm for
computing the spanner.

\section{Preliminaries}
\label{sec:prelim}

Let $P$ be a set of points in the two-dimensional Euclidean space. The
Euclidean graph $\cE$ of $P$ is the complete weighted graph embedded in the
plane whose nodes are identified with the points of $P$. We assume that a
coordinate system is associated with the Euclidean plane and thus every
point can be specified by its $x$ and $y$ coordinates. For every pair of
nodes $u$ and $w$, we identify edge $(u,w)$ with the straight line segment
$[uw]$ and associate an edge length equal to the Euclidean distance
$d_2(u,w) = \sqrt{d_x(u,w)^2 + d_y(u,w)^2}$ where $d_x(u,w)$
(resp. $d_y(u,w)$) is the difference between the $x$ (resp. $y$)
coordinates of $u$ and $w$. Given a constant $t \geq 1$, we say that a
subgraph $H$ of a graph $G$ is a \emph{$t$-spanner}, or simply a {\em spanner},
of $G$ if for any pair of vertices $u,v$ of $G$, the distance between $u$ and
$v$ in $H$ is at most $t$ times the distance between $u$ and $v$ in $G$; the
constant $t$ is referred to as the \emph{stretch factor} of $H$ (with respect
to $G$).  We will say that $H$ is a $t$-spanner, or simply a spanner, if it is
a $t$-spanner of $\cE$.

In the introduction we defined the $L_1$-Delaunay triangulation
as the dual of the Voronoi diagram based on the $L_1$-distance defined as
$d_1(u,w) = d_x(u,w) + d_y(u,w)$ for two points $u$ and $w$.
In this paper, our working definition is an alternate but equivalent one.
Let a {\em square} in the plane be a square whose sides are
parallel to the $x$ and $y$ axes and let a {\em tipped square} be a square
tipped at $45^\circ$. For every pair of points $u,v \in P$, $(u,v)$ is an
edge in the {\em $L_1$-Delaunay triangulation} of $P$ iff there is a tipped
square that has $u$ and $v$ on its boundary and is {\em empty} (i.e.,
it contains no point of $P$ in its interior.
The assumption in this definition is that the points of
$P$ are in {\em general position} which implies that no four points lie on the
boundary of a tipped
square. With this assumption, an $L_1$-Delaunay triangulation is indeed
a plane graph whose interior faces are all triangles.

If a square with sides parallel to the $x$ and $y$ axes, rather than
a tipped square, is used in the above definition then a different triangulation
is obtained; it corresponds to the dual of the Voronoi diagram based on the
$L_\infty$-distance $d_\infty(u,w) = \max\{d_x(u,w), d_y(u,w)\}$. Here again
the assumption is that points of $P$ are in {\em general position} which in
{\em this} case implies that no four points lie
on the boundary of a square. We refer to the resulting triangulation as
the $L_\infty$-Delaunay triangulation. This triangulation is
nothing more than the $L_1$-triangulation of the set of points $P$ after
rotating all the points by $45^\circ$ around the origin. Therefore
Chew's bound of $\sqrt{10}$ on the stretch factor of the $L_1$-Delaunay
triangulation (\cite{Che86})
applies to $L_\infty$-Delaunay triangulations as well. In the remainder
of this paper, we will be using $L_\infty$-Delaunay (rather than $L_1$-)
triangulations because we will be (mostly) using the $L_\infty$-distance,
and squares rather than tipped squares.

In order to avoid technical difficulties we make the usual assumption that
points
of $P$ are in {\em general position} which for us means that 1) no four points
lie on the boundary of a square and 2) no two points have the same $x$ or $y$
coordinate. Note that it is always possible to perturb the points slightly so
they end up in general position and so that a plane spanner on the
perturbed points corresponds to a plane spanner on the original points.
Therefore, the main result in this paper holds for all sets of points and not
just for points in general position.

\section{A Yao subgraph of the $L_\infty$-Delaunay triangulation}
\label{se:yao}

In this section we describe the first step in the construction of our spanner
of $\cE$, the complete Euclidean graph on a set of points $P$. The result
of the first step is a version of the Yao subgraph of $\cE$ on four cones and
first defined by Bose {\it et al.}~\cite{BDD10}.

A \emph{cone} is the open region in the plane between two rays that emanate
from the
same point. With every point $u$ of $P$ we associate four disjoint $90^\circ$
cones emanating from $u$: they are defined by the translation of the $x$-
and $y$-axis from the origin to point $u$ and exclude
the translated axes. We label the cones 0, 1, 2, and 3,
in counter-clockwise order, starting with the cone corresponding to the
first quadrant. Given a cone $i$, the counter-clockwise next cone is cone
$i+1$, whereas the clockwise next cone is cone $i-1$; we assume that arithmetic
on the labels is done modulo 4 so that cone $i+1$ and cone $i-1$ are well
defined. Our general position assumption ensures that no point lies on the
boundary of another point's cone.


Given two points $v$ and $w$, we define $R(v,w)$ to be the rectangle,
with sides parallel to the $x$ and $y$ axes, having $v$ and $w$ as vertices.
The rectangle has positive area because (of our general assumption that)
no two points share the same $x$ or $y$ coordinate.
For a point $v$ and cone $i$ of $v$, we denote by $S_v^i(s)$ the
$s \times s$ square having $v$ as a vertex and whose two sides match the
boundary of cone $i$ of $v$, and by $S_v^i$ the square $S_v^i(s)$ with the
largest $s$ that contains no points of $P$ in its interior
(see Figure~\ref{fi:yao}).

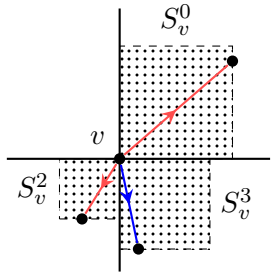
\begin{figure}[!b]
\begin{center}
\begin{tikzpicture}
\draw [thick] (-1.5,0) -- (2,0) node[right] {};
\draw [thick] (0,-1.5) -- (0,2) node[above] {};
\draw (0,0) node[fill,circle,inner sep= 1.5pt,label=135:{$v$}] (u) {};;

\draw (.75,.75) node[draw,dashed,pattern=dots,minimum width = 1.5cm, minimum height=1.5cm,label=90:{$S_v^0$}] {};
\draw (1.5,1.3) node[fill,circle,inner sep = 1.5pt] (v0) {};;
\draw [yaoEdge, color=red!70] (u) -- (v0);

\draw (-.4,-.4) node[draw,dashed,pattern=dots,minimum width = .8cm, minimum height=.8cm,label=180:{$S_v^2$}] {};
\draw (-.5,-.8) node[fill,circle,inner sep = 1.5pt] (v2) {};;
\draw [yaoEdge, color=red!70] (u) -- (v2);

\draw (.6,-.6) node[draw,dashed,pattern=dots,minimum width = 1.2cm, minimum height=1.2cm,label=0:{$S_v^3$}] {};
\draw (0.25,-1.2) node[fill,circle,inner sep = 1.5pt] (v3) {};;
\draw [yaoEdge, color=blue] (u) -- (v3);


\end{tikzpicture}
\caption{Definition of $S_v^i$ and orientation of edges in $\oY4$.}
\label{fi:yao}
\end{center}
\end{figure}

The following is the first step of our spanner construction:
\begin{step}
\label{st:yao}
For every node $v$ of $P$, we choose in each non-empty cone of $v$ the shortest
edge of $\cE$ incident to $v$ according to the $L_\infty$-distance, 
breaking ties arbitrarily, and we give it an orientation out of $v$. 
(See Figure~\ref{fi:yao}.)
\end{step}

We name the resulting directed graph $\oY4$ and denote an edge of $\oY4$ from
node $v$ to node $w$ using notation $\edge{v,w}$ (see 
Figure~\ref{fi:bigexample}-(a)). If edge $\edge{v,w}$ is in $\oY4$ then
$w$ must lie on the boundary of $S_v^i$ for some cone $i$. Because for
every $\edge{v,w} \in \oY4$ there is an empty square with $v$ and
$w$ on its boundary, $(v,w)$ must be an edge in the $L_\infty$-Delaunay
triangulation $T$ of the points in $P$ (see Figure~\ref{fi:bigexample}-(a)). 
Thus the undirected graph obtained by removing the orientations of edges in
$\oY4$ is a subgraph of $T$ which we denote as $\Y4$ (just as in~\cite{BDD10}).
For a given edge $(v,w) \in \Y4$ it is possible that orientation
$\edge{v,w} \in \oY4$, that orientation $\edge{w,v} \in \oY4$, or
that both orientations are in $\oY4$. We will call $(v,w)$ {\em uni-directional}
in the first two cases and {\em bi-directional} in the third case. 

If an edge $(u,v)$ of $\Y4$ is in cone $i$ of $u$ then it must also be
in cone $i+2$ of $v$. One possibility is that $(u,v)$ is the only edge incident
to $u$ in its cone $i$ {\em and} the only edge incident to $v$ in its cone
$i+2$; we call such an edge a {\em mutually-single} edge and note that it
must be bi-directional. If that is not the case, there 
must be either two or more edges of $\Y4$ incident to $u$ in its cone $i$, or
two or more edges of $\Y4$ incident to $v$ in its cone $i+2$, or both.
We call edge $(u,v)$ {\em dual} if there are
two or more edges of $\Y4$ incident to $u$ in its cone $i$ {\em and}
two or more edges of $\Y4$ incident to $v$ in its cone $i+2$.
Finally, given a node $u$ and cone $i$ of $u$, we define the
{\em fan} of $u$ in cone $i$ to be the sequence, in counter-clockwise order,
of all edges of $\Y4$ incident to $u$ in its cone $i$.

\begin{figure}
\begin{center}
\begin{tikzpicture}
\draw [thick] (-.5,0) -- (3.1,0) node[right] {};
\draw [thick] (0,-.5) -- (0,3.1) node[above] {};

\draw (0,0) node[fill,circle,inner sep= 1.5pt,label=225:{$u$}] (u) {};;

\fill [draw, dashed, pattern=dots] (0,2.9)--(.7,2.9)--(.7,2.5)--(1.5,2.5)--(1.5,1.35)--(2,1.35)--(2,.9)--(2.9,.9)--(2.9,0)--(0,0)--(0,2.9);

\draw (.7,2.9) node[fill,circle,inner sep = 1.5pt,label=0:{$v_4$}] (v0) {};
\draw [thick, color=red!70] (v0) -- (u);

\draw (1.5,2.5) node[fill,circle,inner sep = 1.5pt,label=0:{$v_3$}] (v1) {};
\draw [yaoEdge, color=red!70] (v1) -- (u);

\draw (2,1.35) node[fill,circle,inner sep = 1.5pt,label=0:{$v_2$}] (v2) {};
\draw [yaoEdge, color=red!70] (v2) -- (u);

\draw (2.9,.9) node[fill,circle,inner sep = 1.5pt,label=0:{$v_1$}] (v3) {};
\draw [thick, color=red!70] (v3) -- (u);


\draw (1,-.75) node {(a)};
\end{tikzpicture}
\begin{tikzpicture}
\draw [thick] (-.5,0) -- (3.1,0) node[right] {};
\draw [thick] (0,-.5) -- (0,3.1) node[above] {};

\draw (0,0) node[fill,circle,inner sep= 1.5pt,label=225:{$u$}] (u) {};;

\fill [draw, dashed, pattern=dots] (0,2.9)--(1.5,2.9)--(1.5,2.5)--(2,2.5)--(2,1.35)--(2.9,1.35)--(2.9,0)--(0,0)--(0,2.9);

\draw (.7,2.9) node[fill,circle,inner sep = 1.5pt,label=45:{$v_4$}] (v0) {};
\draw [thick, color=red!70] (v0) -- (u);

\draw (1.5,2.5) node[fill,circle,inner sep = 1.5pt,label=45:{$v_3$}] (v1) {};
\draw [yaoEdge, color=red!70] (v1) -- (u);

\draw (2,1.35) node[fill,circle,inner sep = 1.5pt,label=45:{$v_2$}] (v2) {};
\draw [yaoEdge, color=red!70] (v2) -- (u);

\draw (2.9,.9) node[fill,circle,inner sep = 1.5pt,label=0:{$v_1$}] (v3) {};
\draw [thick, color=red!70] (v3) -- (u);


\draw (1,-.75) node {(b)};
\end{tikzpicture}
\begin{tikzpicture}
\draw [thick] (-.5,0) -- (3.1,0) node[right] {};
\draw [thick] (0,-.5) -- (0,3.1) node[above] {};

\draw (0,0) node[fill,circle,inner sep= 1.5pt,label=225:{$u$}] (u) {};;

\draw (.7,2.9) node[fill,circle,inner sep = 1.5pt,label=90:{$v_4$}] (v0) {};
\draw (1.5,2.5) node[fill,circle,inner sep = 1.5pt,label=45:{$v_3$}] (v1) {};

\draw [yaoEdge, color=blue] (v0) -- (v1);
\draw (1.1,2.5) node[draw,dashed,pattern=dots,minimum width = .8cm, minimum height=.8cm,label=180:{$S_{v_4}^3$}] {};

\draw (2,1.35) node[fill,circle,inner sep = 1.5pt,label=45:{$v_2$}] (v2) {};

\draw [yaoEdge, color=blue] (v2) -- (v1);
\draw (1.425,1.925) node[draw,dashed,pattern=dots,minimum width = 1.15cm, minimum height=1.15cm,label=190:{$S_{v_2}^1$}] {};

\draw (2.9,.9) node[fill,circle,inner sep = 1.5pt,label=0:{$v_1$}] (v3) {};

\draw [yaoEdge, color=blue] (v2) -- (v3);
\draw (2.45,0.9) node[draw,dashed,pattern=dots,minimum width = .9cm, minimum height=.9cm,label=190:{$S_{v_2}^3$}] {};

\draw (1,-.75) node {(c)};
\end{tikzpicture}
\end{center}
\caption{Illustrations for Observation~\ref{ob:initial}.}
\label{fi:canonical2}
\end{figure}
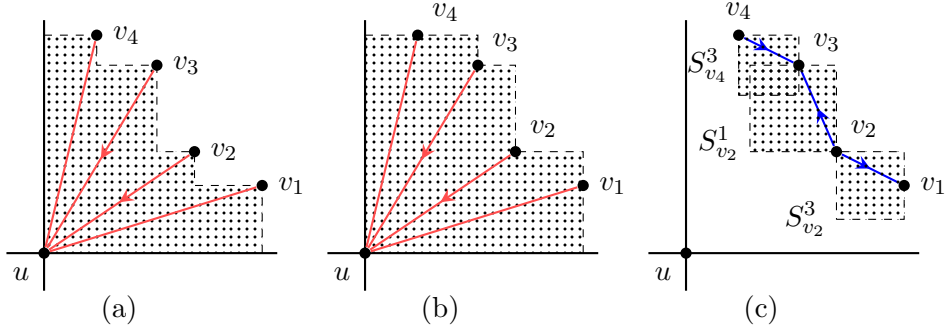

\begin{observation}
\label{ob:initial}
For every node $u$ with a fan $(u,v_1), \dots, (u,v_k)$, with $k \geq 2$, 
in its cone $i$:
\begin{enumerate}[(a)]

\item $R(u,v_l)$ is empty, for
every $l = 1, 2, ..., k$ (see Figure~\ref{fi:canonical2}-(a)).

\item For every $l \in \{2,\dots,k-1\}$, edge $(u,v_l)$ is the only
edge incident to $v_l$ in its cone $i+2$ and $\edge{v_l,u} \in \oY4$
(see Figure~\ref{fi:canonical2}-(a)).

\item For $l \in \{1,k\}$, if edge $(u,v_l)$ is not dual then $(u,v_l)$ is
the only edge incident to $v_l$ in its cone $i+2$
and $\edge{v_l,u} \in \oY4$.

\item For every $l \in \{1, \dots, k-1\}$, $v_l$ lies in cone $i+3$ of
$v_{l+1}$, $v_{l+1}$ lies in cone $i+1$ of $v_l$, and $R(v_l,v_{l+1})$
is empty (see Figure~\ref{fi:canonical2}-(b)).


\item For every $l \in \{1, ..., k-1\}$, 
{\bf if} $d_1(v_l,u) \leq d_1(v_{l+1},u)$
(resp. $d_1(v_l,u) \geq d_1(v_{l+1},u)$) {\bf then} $\edge{v_l, v_{l+1}} \in \oY4$
(resp. $\edge{v_{l+1},v_l} \in \oY4$); furthermore, if edge $(v_l, v_{l+1})$ is
uni-directional then the converse is also true
(see Figure~\ref{fi:canonical2}-(c)).
\end{enumerate}
\end{observation}

\begin{proof}
Since $(u,v_l) \in \Y4$, there is an empty square with $u$ and $v_l$ on its
boundary. Rectangle $R(u,v_l)$ is
contained in this square and thus has no point of $P$ in its interior which
proves part~{\em (a)}.

For every $l \in \{2,\dots,k-1\}$, since $R(u, v_l)$ is empty, 
any edge other than $(u, v_l)$ incident to $v_l$ in its cone $(i+2)$ must either
intersect edge $(u, v_{l-1})$ or $(u, v_{l+1})$, contradicting the planarity of
$\Y4$ (recall that $\Y4$ is a subgraph of the $L_\infty$-Delaunay triangulation
$T$). Thus, for every
$l \in \{2,\dots,k-1\}$, $(u,v_l)$ is the only edge
incident to $v_l$ in its cone $i+2$. By construction of $\oY4$,
$\edge{v_l,u}$ must be in $\oY4$ and $S_{v_l}^{i+2}$ must have $u$ on its
boundary which proves part {\em (b)}. The previous statement is also true for
$l \in \{1,k\}$ if $(u,v_l)$ is the only edge incident to $v_l$ in its cone
$i+2$, which proves part {\em (c)}. 

Since edge $(u,v_{l+1})$ is counter-clockwise from edge $(u,v_l)$ inside cone
$i$ of $u$, for every $l \in \{1,\dots,k-1\}$,
node $v_{l+1}$ cannot be in cone $i+3$ of $v_l$. Node $v_{l+1}$
cannot be in cone $i$ of $v_l$ because rectangle $R(u,v_{l+1})$ would then
contain node $v_l$. Node $v_{l+1}$ also cannot be in cone $i+2$ of $v_l$ because rectangle $R(u,v_l)$
would then contain point $v_{l+1}$.
So $v_{l+1}$ must be in cone $i+1$ of $v_l$ and thus $v_l$ is in cone $i+3$
of $v_{l+1}$. If rectangle $R(v_l,v_{l+1})$ contains a point of $P$ for some
$l=1,\dots,k-1$, let $w$ be a point inside $R(v_l,v_{l+1})$ such that
rectangle $R(w,u)$ is empty. Because $u$ lies inside cone $i+2$ of $w$,
by construction of $\oY4$ there must be an edge in $\oY4$ out of $w$ in
its cone $i+2$. Because $R(w,u)$ is empty, any edge incident to $w$ in its
cone $i+2$ whose endpoint is not $u$ would have to intersect 
edge $(u, v_l)$ or $(u, v_{l+1})$, contradicting the planarity of
$\Y4$. Therefore, $(u,w)$ would have to be an edge of $\Y4$ lying between
$(u,v_l)$ and $(u,v_{l+1})$ in cone
$i$ of $u$ which contradicts our assumption and proves part~{\em (d)}.

Let $d_1(v_l,u) \leq d_1(v_{l+1},u)$ for some $l \in \{1, ..., k-1\}$ (the case
$d_1(v_l,u) \geq d_1(v_{l+1},u)$ is symmetric). We assume first that 
$u$ lies on the boundary of $S_{v_{l+1}}^{i+2}$; from the above proof of
part~{\em (b)}, that is not the case only if $l=k-1$ and $(u,v_{l+1}=v_k)$ is
dual. Since square $S = S_{v_l}^{i+1}(d_\infty(v_l,v_{l+1}))$ lies inside
$R(v_l,v_{l+1}) \cup S_{v_{l+1}}^{i+2}$, $S$ must be empty and no point of $P$
other than $v_l$ and $v_{l+1}$ lies on the boundary of $S$. If $l = k-1$,
$(u,v_{l+1}=v_k)$, but $u$ does not lie on the boundary of $S_{v_k}^{i+2}$
then, by construction of $\oY4$, $v_k$ must lie on the boundary of $S_u^i$.
In that case, $S = S_{v_l}^{i+1}(d_\infty(v_l,v_{l+1}))$ lies inside 
$R(v_l,v_{l+1}) \cup S_u^i$ and so $S$ is empty and no point of $P$
other than $v_l$ and $v_{l+1}$ lies on the boundary of $S$. Since $v_{l+1}$ lies
on the boundary of $S$, square $S_{v_l}^{i+1}$ is exactly square $S$.
Since no point of $P$ other than $v_l$ and $v_{l+1}$ lies on the boundary of
$S_{v_l}^{i+1}$, $\edge{v_l,v_{l+1}} \in \oY4$ and {\em (e)} follows.
\end{proof}

Let $(u,v_1), (u,v_2), ..., (u,v_k)$, with $k \geq 2$, be the fan of $u$ 
in its cone $i$. We call $(u,v_1)$ and $(u,v_k)$ the {\em first} and {\em last}
edge, respectively, in cone $i$ of $u$.\footnote{The first and the last edge in
a cone of $u$ are defined only if there are two or more edges incident to $u$ in
the cone.} We call any remaining edge $(u,v_l)$ ($1 < l < k$) a {\em middle
edge of $u$}; we say that an edge is a {\em middle edge} if it is a
{\em middle edge} of one of its endpoints.
By Observation~\ref{ob:initial}{\em-(e)}, $(v_1,v_2)$, $(v_2,v_3)$, ...,
$(v_{k-1},v_k)$ are all edges
in $\Y4$. We call these edges {\em canonical edges of $u$ in its cone $i$} and
we say that an edge is {\em canonical} if it is a canonical edge of some
node. We make a few observations to differentiate middle,
dual, and canonical edges:
\begin{observation}
\label{ob:initial1.5}
Let $(u,v)$ be an edge of $\Y4$ lying in cone $i$ of $u$.
\begin{enumerate}[(a)]

\item If $(u,v)$ is dual then $(u,v)$ is the first edge in cone $i$ of $u$
and cone $i+2$ of $v$ or the last edge in cone $i$ of $u$ and cone $i+2$ of $v$.


\item If $(u,v)$ is a uni-directional canonical edge such that
$\edge{v,u} \in \oY4$ then $(u,v)$ is the first or last edge in cone $i$
of $u$, the only edge in cone $i+2$ of $v$, and a canonical edge of just one
node.

\item $(u,v)$ can belong to at most one of the following categories: middle,
dual, or canonical.
\end{enumerate}
\end{observation}

\begin{proof}
By Observation~\ref{ob:initial}{\em-(b)}, if $(u,v)$ is dual then it must be
the first or last edge in cone $i$ of $u$ and cone $i+2$ of $v$. W.l.o.g. we
assume that $(u,v)$ is first in cone $i$ of $u$. Since $R(v,u)$ is empty
(Observation~\ref{ob:initial}{\em-(a)}) and because $\Y4$ is planar,
any edge incident to $v$ in its cone $i+2$ other than $(v,u)$ must be
counter-clockwise from $(v,u)$ in cone $i+2$ of $v$. Thus $(v,u)$ must also
be first in cone $i+2$ of $v$ which proves part {\em (a)}. 

If $(u,v)$ is a canonical edge of some node $w$ then, by
Observation~\ref{ob:initial}{\em-(d)}, $w$ must lie in cones $i-1$ of $u$
and $v$ or in cones $i+1$ of $u$ and $v$. W.l.o.g. we assume the former.
If $(u,v)$ is also uni-directional and $\edge{v,u} \in \oY4$ then there must
be at least one more edge of $\Y4$ incident to $u$ in its cone $i$. Let 
$(u,v_1), \dots, (u,v_k)$, with $k \geq 2$, be the fan of $u$ in its cone $i$.
If $v = v_l$ for some $l > 1$ then either $v_1$ is contained in $R(v,w)$,
which contradicts Observation~\ref{ob:initial}{\em-(a)}, or $(v,w)$ intersects
$(u,v_1)$, which contradicts the planarity of $\Y4$. Therefore, $v$ must be
$v_1$. This same argument can be used to show that $(u,v)$ cannot be a canonical
edge of a node in cones $i+1$ of $u$ and $v$. Thus $(u,v)$ is a canonical edge
of node $w$ only. Because, by Observation~\ref{ob:initial}{\em-(a)}, $R(u,w)$
and $R(u,v)$ are empty, any edge incident to $v$ in its cone $i+2$ other than
$(u,v)$ would have
to intersect $(u,v_2)$ or $(u,w)$, which contradicts the planarity of $\Y4$.
Therefore, $(u,v)$ is the only edge in cone $i+2$ of $v$, which completes the
proof of part {\em (b)}.

If $(u,v)$ is a middle edge then, by Observation~\ref{ob:initial}{\em-(b)},
it cannot be the first or the last edge in cone $i$ of $u$ and in cone $i+2$
of $v$. This, together with parts {\em (a)} and {\em (b)}, proves part
{\em (c)}.
\end{proof}

\begin{figure}
\begin{center}
\begin{tikzpicture}
\draw (0.2,-1) node[fill,circle,inner sep = 1.5pt,label=270:{$u$}] (u) {};
\draw (0.8,1) node[fill,circle,inner sep = 1.5pt,label=90:{$v$}] (v) {};
\draw [thick, dashed, color=red!70] (u) -- (v);

\draw (.7,0) node[draw,loosely dotted,minimum width = 2cm, minimum height=2cm,label=180:{$S$}] {};

\draw (.95,-.25) node[draw,dashed,minimum width = 1.5cm, minimum height=1.5cm,label=-60:{$S_u^0$}] {};

\draw (1.1,.4) node[draw,densely dotted,minimum width = 1.2cm, minimum height=1.2cm,label=60:{$S_w^1$}] {};

\draw (1.7,-.2) node[fill,circle,inner sep = 1.5pt,label=0:{$w$}] (w) {};
\draw [yaoEdge, color=red!70] (u) -- (w);
\draw [yaoEdge, color=blue] (w) -- (v);

\end{tikzpicture}
\caption{Proof of Observation~\ref{ob:initial2}.}
\label{fi:yaospanner}
\end{center}
\end{figure}
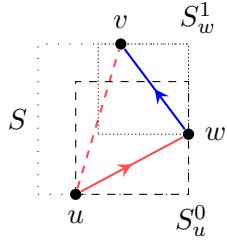

\begin{lemma}
\label{ob:initial2}
$\Y4$ is a plane $(1+\sqrt{2})$-spanner of the $L_\infty$-Delaunay
triangulation $T$ and also a $(1+\sqrt{2})\sqrt{4+2\sqrt{2}}$-spanner of $\cE$.
\end{lemma}
\begin{proof}
Let $(u,v)$ be an edge in $T$ that is not in $\Y4$
and let $(u,v)$ lie in cone $i$ of $u$.
Since $(u,v) \in T$, there exists a square $S$ circumscribing
$(u,v)$ and containing no points of $P$ in its interior. W.l.o.g, we assume
that $v$ is in cone 0 of $u$. If $u$ and $v$ lie on adjacent sides of
$S$ then either $S_u^0$ is contained in $S$ and has $v$ on its boundary,
implying $\edge{u,v} \in \oY4$, or $S_v^2$ is contained in $S$ and has
$u$ on its boundary, implying $\edge{v,u} \in \oY4$, contradicting the
assumption that $(u,v) \not\in \Y4$. Therefore we can assume that $u$
and $v$ lie in the interior of opposite sides of
the square $S$, say the bottom and top sides, respectively. We can also
assume that $S$ also has another point, $w$,
on its boundary, say on the interior of the right side of $S$ (otherwise we
translate $S$ to the right until that occurs) as shown in
Figure~\ref{fi:yaospanner}. Because $S$ is devoid of points of $P$,
edges $(u,w)$ and $(w,v)$ are in $T$. Since $d_x(u,w)$ and $d_y(u,w)$ are
both less than $d_y(u,v)$, it follows that
$d_\infty(u,w) < d_\infty(u,v)$. If $d_x(u,w) \geq d_y(u,w)$ then
square $S_u^0(d_x(u,w))$ is contained inside square $S$ and is thus empty
and $\edge{u,w} \in \oY4$; otherwise
square $S_w^2(d_y(u,w))$ is empty and $\edge{w,u} \in \oY4$. So $(u,w) \in \Y4$.
A similar argument can be used to show that $(w,v)$ is in $\Y4$ as well. Given
that $u$ and $v$ lie on the bottom and top sides of square $S$ and $w$ lies on
the right side of $S$, it follows that
$d_2(u,w)+d_2(w,v) \leq (1+\sqrt{2}) d_2(u,v)$.

Since the $L_\infty$-Delaunay triangulation is a $\sqrt{4+2\sqrt{2}}$-spanner of
the complete Euclidean graph~\cite{BGHP12}, it follows that $\Y4$ is a 
$(1+\sqrt{2})\sqrt{4+2\sqrt{2}}$-spanner of $\cE$.
\end{proof}


\section{Anchors and standard paths}
\label{se:anchors}

In cone $i$ of some node $u$, either there is no edge
of $\Y4$ incident to $u$ or there is exactly one edge of $\oY4$ out of $u$ and
any number of edges of $\oY4$ into $u$. In this section we describe how, under
certain conditions, we choose a special {\em anchor} among all those edges.
We then use anchors to define special paths between endpoints of every
edge in $\Y4$. We start with some definitions:

\begin{definition} \rm
\label{de:maxuni}
Let $(u,v_1),\dots,(u,v_k)$ be the fan of $u$ in its cone $i$. For
any $r,s \in \{1,\dots,k\}$, we define $\cpath{u}{v_s,v_r}$ to be the path
$v_s,v_{s+1},\dots,v_r$ (if $s \leq r$) or $v_s,v_{s-1},\dots,v_r$ (if $s > r$)
in $\Y4$. We will say that $\cpath{u}{v_s,v_r}$
is a {\em uni-directional canonical path} if every edge in
the path is uni-directional and $v_s,\dots,v_r$ forms a directed path from
$v_s$ to $v_r$ in $\oY4$. A uni-directional canonical path {\em ending at $v_r$}
is {\em maximal} is it is not contained in any other uni-directional canonical
path ending at $v_r$.
\end{definition}
For instance, in Figure~\ref{fi:bigexample}{\em-(a)}, $(u_9,u_3)$ and $(u_9,u_{20})$ belong to the
fan of $u_9$ hence the path $\cpath{u_9}{u_3,u_{20}}$ is well-defined. The path
$\cpath{u_9}{u_{20},u_{17}}$ is a maximal uni-directional canonical path ending
at $u_{17}$.
Note that if $v_s,v_{s+1},\dots,v_r$ is a {\em maximal uni-directional canonical
path ending at $v_r$} then either $s=1$ or $\edge{v_s,v_{s-1}} \in \oY4$.
Similarly, if $v_s,v_{s-1},\dots,v_r$ is a {\em maximal uni-directional
canonical path ending
at $v_r$} then either $s=k$ or $\edge{v_s,v_{s+1}} \in \oY4$. We can now define
anchor edges:
\begin{definition} \rm 
\label{de:anchors}
For every node $u$ and every cone $i$ of $u$ containing an edge incident to $u$:
\begin{enumerate}[(a)]

\item If $(u,v) \in \Y4$ is a mutually-single edge in cone $i$ of $u$, we
define $(u,v)$ to be the anchor chosen by $u$ in cone $i$.

\item If there are two or more edges of $\Y4$ incident to $u$ in its cone $i$,
let $(u,v_1), \dots, (u,v_k)$, with $k \geq 2$, be the fan of $u$ in cone $i$
and let $\edge{u,v_l}$, for some $l \in \{1,...,k\}$, be the only edge of
$\oY4$ in cone $i$ of $u$ that is outgoing with respect to $u$:
\begin{enumerate}[(i)]

\item If $l \geq 2$ and $\edge{v_{l-1}, v_l} \in \oY4$ but
$\edge{v_l, v_{l-1}} \not\in \oY4$, let $l' < l$ be such that
$\cpath{u}{v_{l'},v_l}$ is a maximal uni-directional canonical path ending
at $v_l$; we define $(u,v_{l'})$ to be the anchor chosen by $u$ in cone $i$.

\item Otherwise, if $l \leq k-1$ and $\edge{v_{l+1}, v_l} \in \oY4$ but
$\edge{v_l, v_{l+1}} \not\in \oY4$, let $l' > l$ be such that
$\cpath{u}{v_{l'},v_l}$ is a maximal uni-directional canonical path ending
at $v_l$; we define $(u,v_{l'})$ to be the anchor chosen by $u$ in cone $i$.

\item Otherwise, we define $(u,v_{l'}) = (u,v_l)$ to be the anchor chosen
by $u$ in cone $i$.

\end{enumerate}
\end{enumerate}

\end{definition}
We use the notation $\anchor_i(u)$ to denote the anchor edge chosen by node
$u$ in its cone $i$. In Figures~\ref{fi:bigexample}{\em-(b)-(d)},
an anchor $(u,v_{l'})$ is represented
by a thick edge with an arrow toward $v_{l'}$ at the end of the edge.
In Figure~\ref{fi:bigexample}{\em-(b)}, $\anchor_0(u_{22})=(u_{22},u_{15})$
illustrates case (a), $\anchor_3(u_9) = (u_9,u_{20})$ illustrates
  case (b)-(i) and $\anchor_3(u_{22}) = (u_{22},u_{24})$ illustrates
  case (b)-(iii).
Note that if there is only one edge $(u,v)$ of $\Y4$
incident to $u$ in its cone $i$ but $v$ has two or more edges of $\Y4$
incident to it in its cone $i+2$, then $\anchor_i(u)$ is not
defined. For instance in Figure~\ref{fi:bigexample}{\em-(b)}, this is the case
for $\anchor_3(u_{14})$ and $\anchor_1(u_7)$. It is always true, however,
that if $(u,v)$ is an edge lying in cone $i$ of $u$ then either $\anchor_i(u)$
or $\anchor_{i+2}(v)$ is defined. We use this to define a special type
of path for every edge $(u,v) \in \Y4$:

\begin{definition} \rm
Let $(u,v)$ be an edge of $\Y4$ lying in cone $i$ of $u$ such that
$\anchor_i(u) = (u,v')$ is defined. The {\em $1$-standard path} from $u$ to $v$
consists of edge $(u,v')$ together with $\cpath{u}{v,v'}$.
\end{definition}
Since $\cpath{u}{v,v'} \in \Y4$, there is a $1$-standard path from $u$ to $v$ or
from $v$ to $u$ for every
edge $(u,v) \in \Y4$. If $(u,v)$ is a dual edge in $\Y4$, there is a
$1$-standard path from $v$ to $u$ as well as one from $u$ to $v$. The same
is true if $(u,v)$ is a mutually-single edge in cone $i$ of $u$.
\begin{lemma}
\label{le:L1path}
Let $(u,v)$ be an edge of $\Y4$ lying in cone $i$ of $u$ such that
$\anchor_i(u) = (u,v')$ is defined. Then (as illustrated in
Figure~\ref{fi:path}-(b)):
\begin{enumerate}[(a)]

\item $d_2(u,v') \leq 2 d_2(u,v)$,

\item the length of any edge in $\cpath{u}{v,v'}$ is at most
$\sqrt{2}d_2(u,v)$, and

\item the length of $\cpath{u}{v,v'}$ is at most $(1+\sqrt{2})d_2(u,v)$.
\end{enumerate}
The length of the $1$-standard path from $u$ to $v$ is thus at most
$(3+\sqrt{2})d_2(u,v)$.
\end{lemma}

\begin{figure}[!b]
\begin{center}
\begin{tikzpicture}
\draw [thick] (-.5,0) -- (3.6,0) node[right] {};
\draw [thick] (0,-.5) -- (0,3.6) node[above] {};
\draw (0,0) node[fill,circle,inner sep= 1.5pt,label=225:{$u$}] (u) {};;

\draw (.25,3) node[fill,circle,inner sep = 1.5pt,label=45:{$v=v_r$}] (vr) {};
\draw [yaoEdge,color=gray] (vr) -- (u);

\draw (1.2,2.7) node[fill,circle,inner sep = 1.5pt] (2) {};
\draw [yaoEdge, color=blue] (vr) -- (2);
\draw [yaoEdge,color=gray] (2) -- (u);

\draw (1.5,2) node[fill,circle,inner sep = 1.5pt,label=45:{$v_l$}] (vlpr) {};
\draw [yaoEdge, color=blue] (vlpr) -- (2);
\draw [yaoEdge,color=gray] (u) -- (vlpr);

\draw (2.1,1.2) node[fill,circle,inner sep = 1.5pt] (3) {};
\draw [yaoEdge, color=blue] (3) -- (vlpr);
\draw [yaoEdge,color=gray] (3) -- (u);

\draw (2.5,0.7) node[fill,circle,inner sep = 1.5pt] (4) {};
\draw [yaoEdge, color=blue] (4) -- (3);
\draw [yaoEdge,thick,color=gray] (4) --(u);

\draw (2.9,.15) node[fill,circle,inner sep = 1.5pt,label=-90:{$v'=v_{l'}$}] (vl) {};
\draw [yaoEdge,thick, color=red!70] (vl) -- (u);
\draw [anchorEdge, color=red!70] (u) -- (vl);
\draw [yaoEdge, color=blue] (vl) -- (4);



\draw [dotted] (vlpr)--(3.5,0);
\draw [dotted] (3)--(3.3,0);
\draw [dotted] (4)--(3.2,0);
\draw [dashed] (0,2)--(2,2)--(2,0);
\end{tikzpicture}\hspace{1cm}
\begin{tikzpicture}
\draw [thick] (-.5,0) -- (3.6,0) node[right] {};
\draw [thick] (0,-.5) -- (0,3.6) node[above] {};
\draw (0,0) node[fill,circle,inner sep= 1.5pt,label=225:{$u$}] (u) {};;

\draw (0,2) node[fill,circle,inner sep = 1.5pt,label=45:{$v=v_r$}] (vr) {};

\draw (2,2) node[fill,circle,inner sep = 1.5pt,label=45:{$v_l$}] (vlpr) {};
\draw [color=blue] (vlpr) -- (vr);

\draw (4, 0) node[fill,circle,inner sep = 1.5pt,label=-90:{$v'=v_{l'}$}] (vl) {};
\draw [yaoEdge, thick, color=red!70] (vl) -- (u);
\draw [anchorEdge, color=red!70] (u) -- (vl);
\draw [yaoEdge, color=blue] (vl) -- (vlpr);
\draw [yaoEdge, color=gray] (u) -- (vlpr);
\draw [dashed] (0,2)--(2,2)--(2,0);
\end{tikzpicture}

(a) \hspace{2in} (b)
\end{center}
\caption{(a) The $1$-standard path from $u$ to $v=v_r$ consists of
$\anchor_i(u) = (u,v'=v_{l'})$ (in red) and $\cpath{u}{v'=v_{l'},v=v_r}$
(in blue).
Illustrated is the case when $\edge{u,v_l} \in \oY4$ for $r > l > l'$;
$\cpath{u}{v'=v_{l'},v_l}$ is a uni-directional canonical path. (b) Illustration
of Lemma~\ref{le:L1path}.}
\label{fi:path}
\end{figure}
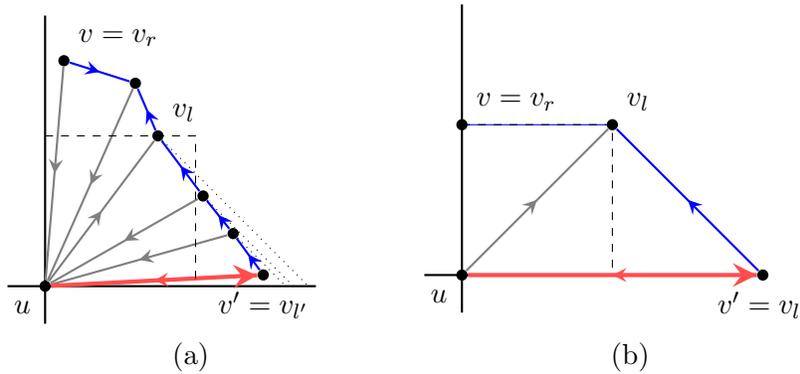

\begin{proof}
We assume w.l.o.g. that $i = 0$. If $(u,v) = (u,v')$ the lemma trivially holds.
Otherwise, let $(u,v_1), \dots, (u,v_k)$ be the
fan of $u$ in its cone $0$ and let $\edge{u,v_l} \in \oY4$, $v' = v_{l'}$,
and $v = v_r$ for some $l,l',r \in \{1, \dots, k\}$. We assume w.l.o.g. that
$r > l'$.
We assume that $r > l > l'$ as shown in Figure~\ref{fi:path}-(a) and
Figure~\ref{fi:path}-(b).
The proof for this case can be applied to prove the cases when $l \geq r$ or
$l' \geq l$. By Definition~\ref{de:anchors},
$\cpath{u}{v_{l'},v_l}$ is a maximal uni-directional canonical path ending at
$v_l$ and, using Observation~\ref{ob:initial}{\em-(e)}, we have
$d_2(u,v_{l'}) \leq d_1(u,v_{l'}) \leq d_1(u,v_l) \leq 2d_\infty(u,v_l)$.
Since $(u,v_l)$ is the shortest (with respect to the $L_\infty$-distance)
edge in cone $0$ of $u$, we have
$2d_\infty(u,v_l) \leq 2d_\infty(u,v_r) \leq 2d_2(u,v_r)$.
which proves part {\em (a)}.

Any edge $(v_s,v_{s+1})$, for $s \in \{l', \dots, l-1\}$, of
$\cpath{u}{v_{l'},v_l}$ must lie within rectangle $R(v_s,v_{s+1})$ and thus
$d_2(v_s,v_{s+1}) \leq \sqrt{2}d_\infty(v_s,v_{s+1})$. Since
$\cpath{u}{v_{l'},v_l}$ is a maximal uni-directional canonical path ending at
$v_l$ and using Observation~\ref{ob:initial}{\em-(e)}, we have that
$d_\infty(v_s,v_{s+1}) = d_y(v_s,v_{s+1})$ which in turn is at most
$d_y(v_{l'},v_l) \leq d_y(u,v_l) \leq d_\infty(u,v_l) \leq d_\infty(u,v_r) \leq d_2(u,v_r)$.
Thus $d_2(v_s,v_{s+1}) \leq \sqrt{2}d_2(u,v_r)$.
Any edge $(v_s,v_{s+1})$ for $s \in \{l, \dots, r-1\}$ must lie within
$R(v_r,v_l)$ which in turn is contained in $S_u^i(d_\infty(u,v_r))$. Therefore
$d_2(v_s,v_{s+1}) \leq \sqrt{2}d_\infty(u,v_r) \leq \sqrt{2}d_2(u,v_r)$ which
completes the proof of part {\em (b)}.

Using an argument similar to the one we used in the previous paragraph,
$d_2(v_s,v_{s+1}) \leq \sqrt{2}d_y(v_s,v_{s+1})$,
and hence, the length of $\cpath{u}{v_{l'},v_l}$ is at most
$\sqrt{2}d_y(v_{l'},v_l)$. The length of $\cpath{u}{v_r,v_l}$ is at most
$d_1(v_r,v_l) = d_x(v_r,v_l) + d_y(v_r,v_l)$. So the length of
$\cpath{u}{v_{l'},v_r}$ is at most $d_x(v_r,v_l) + d_y(v_r,v_l) +
\sqrt{2}d_y(v_{l'},v_l) \leq d_x(u,v_l) + \sqrt{2}d_y(u,v_r) \leq
d_\infty(u,v_l) + \sqrt{2}d_\infty(u,v_r) \leq (1+\sqrt{2})d_\infty(u,v_r) \leq
(1+\sqrt{2})d_2(u,v_r)$, which completes the proof.
\end{proof}

The reader can verify that the graph obtained by taking the union
of $1$-standard paths defined over all edges in $\Y4$, is a
$(3+\sqrt{2})$-spanner of $\Y4$ of maximum degree at most 12. 
We omit the proof because we do not make use of this fact in the
rest of the paper.

\begin{definition} \rm
\label{de:strongweak}
An anchor $(u,v)$ chosen by $u$ in its cone $i$ is {\em strong} if
$\anchor_{i+2}(v) = (v,u)$ or if $\anchor_{i+2}(v)$ is not defined;
it is {\em weak} if $\anchor_{i+2}(v) \not= (v,u)$ (see
Figure~\ref{fi:bigexample}-(b)).
\end{definition}

Let $(u,v)$ be an edge of $\Y4$ lying in cone $i$ of $u$. If $(u,v)$ is
mutually-single then $(u,v)$ is a strong anchor. For
edges that are not mutually-single, we make the following observations:
\begin{observation}
\label{ob:anchors1}
Let $u$ be a node that has a fan $(u,v_1),\dots,(u,v_k)$, with $k \geq 2$,
in its cone $i$. Then $\anchor_i(u)$ is defined, and if
$\anchor_i(u) = (u,v_{l'})$ for some $l' \in \{1, \dots, k\}$ then:
\begin{enumerate}[(a)]

\item If $l'=1$ then $\edge{v_1,v_2} \in \oY4$. If $l'=k$ then
$\edge{v_k,v_{k-1}} \in \oY4$. If $1 < l' < k$ then
$\edge{v_{l'},v_{l'-1}} \in \oY4$ and $\edge{v_{l'},v_{l'+1}} \in \oY4$.


\item If $(u,v_{l'})$ is a weak anchor then it is a dual edge.

\end{enumerate}
\end{observation}

\begin{proof}
By Definition~\ref{de:anchors}{\em-(b)}, if $u$ has a fan of size at least 2
in its cone $i$ then $\anchor_i(u)$ is defined. Part {\em (a)} follows from
Observation~\ref{ob:initial}{\em-(e)} and
Definition~\ref{de:anchors}{\em-(b)}.
For part {\em (b)}, note that if $(u,v_{l'})$ is a
weak anchor then $\anchor_{i+2}(v_{l'})$ is defined and is not $(u,v_{l'})$ which
means that there must be two or more edges of $\Y4$ incident to $v$ in its
cone $i+2$. Therefore $(u,v_{l'})$ is a dual edge.
\end{proof}

By Definition~\ref{de:anchors} and Definition~\ref{de:strongweak}, if
$\anchor_i(w_1) = (w_1,w_2)$ is a weak anchor then
$\anchor_{i+2}(w_2)$ is defined, and $\anchor_{i+2}(w_2)$ could be
either weak or strong. This means that starting from any weak anchor
$(w_1,w_2)$ there is a well-defined path of weak anchors
$\anchor_i(w_1) = (w_1,w_2)$,
$\anchor_{i+2}(w_2) = (w_2,w_3)$, $\anchor_i(w_3) = (w_3,w_4)$, ...
that would end when a strong anchor is encountered.
Furthermore, if $\anchor_i(w_1) = (w_1,w_2)$ is a weak anchor then any
other anchor incident to $w_1$ in its cone $i$ would have
to be weak. Since weak anchors are dual and there can only be two dual
edges incident to $w_1$ in its cone $i$, one of which is $(w_1,w_2)$,
there can be only one other weak anchor incident to $w_1$ in its cone $i$, say
$\anchor_{i+2}(w_0) = (w_0,w_1)$.
By repeatedly applying Observation~\ref{ob:initial}{\em-(d)} to every
successive pairs of nodes on $\cpath{w_j}{w_{j+1},w_{j-1}}$, we note that
$w_{j+1}$ is always in cone $i+3$ of $w_{j-1}$
and so the path $w_0,w_1,\dots$ cannot form a cycle. This means that we can
partition all weak anchors into maximal paths that we define as follows:
\begin{definition} \rm
A \emph{weak anchor chain} is a path
$w_0,w_1,\dots, w_k$ of maximal length consisting, for some $i \in \{0,1,2,3\}$,
of weak anchors $\anchor_i(w_0) = (w_0,w_1)$, $\anchor_{i+2}(w_1) = (w_1,w_2)$,
$\anchor_i(w_2) = (w_2,w_3)$, $\anchor_{i+2}(w_3) = (w_3,w_4)$, and so on until:
\begin{itemize}
\item if $k$ is even, weak anchor $\anchor_{i+2}(v_{k-1}) = (v_{k-1}, v_k)$ such that
$\anchor_i(v_k) = (v_k,w)$ is a strong anchor, or
\item if $k$ is odd, weak anchor $\anchor_i(v_{k-1}) = (v_{k-1}, v_k)$ such that
$\anchor_{i+2}(v_k) = (v_k,w)$ is a strong anchor.
\end{itemize}
\end{definition}

For instance, in Figure~\ref{fi:bigexample}{\em-(b)} $u_{22},
u_{24},u_{15},u_{5}$ and $u_{28},u_{26}, u_{27}$ and $u_{21},u_5$ are three
weak anchor chains.
We now select anchor edges that will actually be included in the spanner:
\begin{definition} \rm
\label{de:anchors2}
We designate all strong anchors as {\em selected}. Furthermore,
for every weak anchor chain $w_0, w_1, ..., w_k$:
\begin{enumerate}[(a)]

\item For $l = k-1, k-3, \dots$, we designate anchor $(w_{l-1}, w_l)$
(chosen by $w_{l-1}$) as {\em selected}.


\item If $(w_0,w_1)$ is not selected, i.e. $k$ is odd, then we designate
$(w_0,w_1)$ to be a {\em start-of-odd-chain anchor} (chosen by $w_0$).
\end{enumerate}
\end{definition}
In Figure~\ref{fi:bigexample}-(c), the dashed (not dotted) edges are the
selected weak anchors; start-of-odd-chain anchors include edges
$(u_3,u_{15}),(u_2,u_{10}), (u_{21},u_5),(u_{17},u_4),(u_{22},u_{24})$
and $(u_{13},u_1)$. The following observations regarding anchors are easy to
check and the proofs are left to the reader:
\begin{observation}
\label{ob:anchors2}
For every node $u$ and cone $i$ of $u$:
\begin{enumerate}[(a)]
\item There is at most one selected anchor incident to $u$ in
its cone $i$, whether the anchor is chosen by $u$ or not.
\item If $(u,v)$ is a start-of-odd-chain anchor {\em chosen by $u$} in its cone $i$
then there is no selected anchor incident to $u$ in its cone $i$.
\item If $(u,v)$ is an anchor chosen by $u$ in its cone $i$ that is not
selected, then there is a selected anchor chosen by $v$ in its cone
$i+2$.
\item If $(u,v)$ is an anchor chosen by $u$ in its cone $i$ that is not
selected and that is not a start-off-odd-chain anchor, then there is another,
selected anchor incident to $u$ in its cone $i$ (chosen by a node other than
$u$).
\end{enumerate}
\end{observation}

We now define a new type of standard path that makes use of selected anchors
only:
\begin{definition} \rm
\label{de:2standard}
Let $(u,v)$ be an edge of $\Y4$ lying in cone $i$ of $u$ such that
$\anchor_i(u) = (u,v')$ is defined. The {\em $2$-standard path} from $u$ to
$v$ is:
\begin{itemize}

\item The $1$-standard path from $u$ to $v$, if $\anchor_i(u)$ is
selected.

\item The path $\cpath{u}{v,v'}$ together with the $1$-standard path from $v'$
to $u$, if $\anchor_i(u)$ is not selected.
\end{itemize}
\end{definition}
By Observation~\ref{ob:anchors2}{\em-(c)}, if $\anchor_i(u)=(u,v')$ is defined
but not selected then $\anchor_{i+2}(v')$ is defined and selected. The
$1$-standard path from $v'$ to $u$ is thus well-defined and hence a $2$-standard
path from $u$ to $v$ is well-defined. By applying Lemma~\ref{le:L1path} twice,
the length of the $2$-standard path from $u$ to $v$ is at most
$(3+\sqrt 2)^2 d_2(u,v)$. In Figure~\ref{fi:bigexample}{\em-(a)}, the
$2$-standard path from $u_{22}$ to $u_{23}$ is path
$u_{22},u_{15},u_{24},u_{18},u_{23}$. Note that non-selected anchors do not appear
in $2$-standard paths.

In the next section we will construct our first bounded degree spanner of
$\Y4$. We will show that it contains short and highly structured paths between
the endpoints of every edge in $\Y4$. We define the structure of these paths
now:

\begin{definition} \rm
\label{de:standard}
Let $H$ be a subgraph of $\Y4$ that includes all selected edges, let
$(u,v)$ be an edge of $\Y4$ lying in cone $i$ of $u$ such that
$\anchor_i(u)$ is defined, and let $p$ be the $2$-standard path from $u$ to
$v$. For every $d \geq 1$, the {\em $2d$-standard pre-path in $H$} from $u$
to $v$ consists of:
\begin{enumerate}[(a)]

\item all edges on $p$ that are in $H$

\item and, if $d > 1$, the $2(d-1)$-standard pre-path in $H$ from $w$ to $w'$ or
from $w'$ to $w$ for every canonical edge $(w, w')$ on $p$ that is not in $H$.
\end{enumerate}
When the $2d$-standard pre-path in $H$ is a path from $u$ to $v$ we call it the
{\em $2d$-standard path} or, more simply, the {\em standard path} when the value
of $d$ is understood from the context.
\end{definition}
Because there is a $2$-standard path for every $(u,v) \in \Y4$ and because $H$
includes selected anchors, $2d$-standard pre-paths in $H$ are well-defined for all
$(u,v) \in \Y4$. When a $2d$-standard pre-path in $H$ from $u$ to $v$ is a path,
its length can be bounded by $(3+\sqrt 2)^{2d} d_2(u,v)$ by applying
Lemma~\ref{le:L1path} recursively.




\section{A spanner of maximum degree at most 8}
\label{se:H8}


We now construct $H_8$, our first bounded degree spanner of $\Y4$. It consists
of all selected anchors and a subset of the uni-directional canonical
edges of $\Y4$. We choose the edges in $H_8$ as follows:
\begin{step}
\label{st:spanner1}
We choose all the selected anchors. Then, for every node $u$ and cone $i$
of $u$, if $(u,v_1), \dots, (u,v_k)$, with $k \geq 2$, is the fan $u$ in its
cone $i$, we choose all the uni-directional edges on $\cpath{u}{v_1,v_k}$
except for the following two cases:
\begin{enumerate}[(a)]

\item $(v_2,v_1) \not\in H_8$ iff $(v_1,u)$ is a dual edge but not a
start-of-odd-chain anchor chosen by $v_1$ and edge $(v_2,v_1)$ is a
non-anchor, uni-directional edge such that $\edge{v_2,v_1} \in \oY4$ but
$\edge{v_1,v_2} \not\in \oY4$.

\item $(v_{k-1},v_k) \not\in H_8$ iff $(v_k,u)$ is a dual edge but
not a start-of-odd-chain anchor chosen by $v_k$ and edge $(v_{k-1},v_k)$ is a
non-anchor, uni-directional edge such that $\edge{v_{k-1},v_k} \in \oY4$ but
$\edge{v_k,v_{k-1}} \not\in \oY4$.
\end{enumerate}
\end{step}

In order to facilitate the analysis of the
maximum degree of the spanner, we devise a {\em charging scheme} that assigns
each chosen spanner edge $(u, v)$ to a cone of $u$ and to a cone of $v$.
A chosen edge $(u,v)$ is charged to the cone of $u$
containing it if it is a selected anchor or if
it is a uni-directional canonical edge with $\edge{u,v} \in \oY4$. By
Observation~\ref{ob:anchors2}{\em-(a)} and
Observation~\ref{ob:initial1.5}{\em-(b)}
at most one edge is charged to a
cone in this way. If, however, edge $(u,v)$ is a non-anchor,
uni-directional canonical edge of some node $w$ with $\edge{v,u} \in \oY4$,
i.e. its orientation in $\oY4$ is incoming at $u$, it is charged to the cone
of $u$ that contains $w$. Therefore, to bound the degree of each
node we only need to focus on cones that have one or more non-anchor,
uni-directional, canonical, incoming (in $\oY4$) edges charged to them.

We show in the following lemma that such cones have at most 2 edges charged
to them, which will imply that $H_8$ has maximum degree at most $8$:

\begin{lemma}
\label{le:charge2}
Let $u$ be a node with the fan $(u,v_1), \dots, (u,v_k)$, with $k \geq 2$,
in its cone $i$ and
let $r \in \{1,\dots,k\}$. No more than 1 edge is charged to cone $i+2$ of
$v_r$ except in the following cases when 2 edges are charged:
\begin{enumerate}[(a)]
\item $1 < r < k$, $(v_r,u) \not\in H_8$ and both $(v_{r-1}, v_r)$ and
$(v_r,v_{r+1})$ are non-anchor uni-directional canonical edges in $H_8$ such that
$\edge{v_{r-1},v_r} \in \oY4$ and $\edge{v_{r+1},v_r} \in \oY4$
(e.g., cone 3 of $u_{14}$ in Figure~\ref{fi:bigexample}-(c)).

\item $r=1$, $(v_1,u)$ is a non-anchor uni-directional canonical edge in $H_8$ such that
$\edge{v_1,u} \in \oY4$ and $(v_2, v_1)$ is a non-anchor uni-directional
canonical edge in $H_8$ such that $\edge{v_2,v_1} \in \oY4$ (e.g., cone 3
of $u_{21}$ Figure~\ref{fi:bigexample}-(c)).

\item $r=k$, $(v_k,u)$ is a non-anchor uni-directional canonical edge in $H_8$ such that
$\edge{v_k,u} \in \oY4$ and $(v_{k-1}, v_k)$ is a non-anchor uni-directional
canonical edge in $H_8$ such that $\edge{v_{k-1},v_k} \in \oY4$ (e.g., cone
0 of $u_{19}$ Figure~\ref{fi:bigexample}-(c)).
\end{enumerate}
\end{lemma}

\begin{proof}
We consider first the case $1 < r < k$.
If an incoming (in $\oY4$) edge is charged to cone $i+2$ of $v_r$
then either $\edge{v_r,v_{r-1}} \not \in \oY4$
or $\edge{v_r,v_{r+1}} \not\in \oY4$. By Observation~\ref{ob:anchors1}{\em-(a)},
this means that $\anchor_i(u) \not= (u,v_r)$. Furthermore, since $(u,v_r)$ is
a middle edge, by Observation~\ref{ob:initial1.5}{\em-(c)}, $(u,v_r)$ is not
canonical and thus could not be chosen in Step~\ref{st:spanner1}.
So $(u,v_r) \not\in H_8$.
Therefore, the maximum charge of 2 in cone $i+2$ of $v_r$
occurs when both $(v_{r-1},v_r)$ and $(v_{r+1},v_r)$ are charged to it
which happens when they are both non-anchor uni-directional canonical edges
with $\edge{v_{r-1},v_r} \in \oY4$, and $\edge{v_{r+1},v_r} \in \oY4$.

We consider case $r=1$ next. Edge $(v_2,v_1)$ is charged to cone $i+2$ of
$v_1$ only if $(v_2,v_1)$ is a non-anchor, uni-directional edge,
$\edge{v_2,v_1} \in \oY4$, and either $(v_1,u)$ is not dual or $(v_1,u)$ is
a start-of-odd-chain anchor chosen by $v_1$.

In the second case, $(v_1,u)$ is not a selected anchor and is
thus not in $H_8$. Furthermore, by Observation~\ref{ob:anchors2}{\em-(b)},
no other selected anchor is also charged to cone $i+2$ of $v_1$.
Finally,
there cannot be another start-of-odd-chain anchor out of $v_i$ in its cone
$i+2$ because there is only one anchor out of a node in a cone. So cone
$i+2$ of $v_1$ could not be charged more than 1 in this case.

If $(v_1,u)$ is not dual then it must be the only edge of $\Y4$ in cone
$i+2$ of $v_1$. Thus, for cone $i+2$ to
have a charge of 2, it must be that $(v_1,u) \in H_8$. Note that $(v_1,u)$
cannot be an anchor chosen by $v_1$ because by Definition~\ref{de:anchors}
$v_1$ chose no anchor. It also cannot be $\anchor_i(u)$ because, using
Observation~\ref{ob:anchors1}{\em-(a)}, that would violate the assumption
that $(v_2,v_1)$ is uni-directional and $\edge{v_2,v_1} \in \oY4$. Therefore
$(v_1,u)$ must be a uni-directional canonical edge (with respect to some node
$w$ in cone $i+3$ of $u$ and $v_1$) and $\edge{v_1,u} \in \oY4$.

The case $r=k$ follows by a symmetric argument to that in the case when $r=1$.
\end{proof}

We will show next that $H_8$ is a spanner of $\Y4$. In fact, we show something
stronger:

\begin{theorem}
\label{th:span}
Let $(u,v)$ be an edge of $\Y4$ lying in cone $i$ of $u$. There is
a $6$-standard path in $H_8$ from $u$ to $v$ or from $v$ to $u$
of length at most $(3+\sqrt 2)^6 \cdot d_2(u,v)$.
\end{theorem}

We start by proving two special cases of the theorem:
\begin{lemma}
\label{le:bidirectional}
Let $(u,v)$ be a bi-directional canonical edge of $\Y4$ lying in cone $i$ of
$u$ such that $\anchor_i(u)$ is defined. The $2$-standard path from
$u$ to $v$ is in $H_8$ and has length at most $(3+\sqrt 2)^2 \cdot d_2(u,v)$.
\end{lemma}

\begin{proof}
W.l.o.g. we assume that $i=0$. The lemma clearly holds if
$\anchor_0(u) = (u,v)$ since $(u,v)$, being a canonical edge, would have to be
a selected anchor and thus in $H_8$. Therefore we can assume that
$\anchor_0(u) \not= (u, v)$ and, w.l.o.g., that $\anchor_0(u)$ is
clockwise from edge $(u, v)$ in cone $0$ of $u$. Then, since $(u,v)$ is
canonical, $(u,v)$ must be the last edge in cone $0$ of $u$. Let
$(u,v_1), \dots, (u,v_k) = (u, v)$ be the fan of $u$ in its cone $0$
and let $\anchor_0(u) = (u,v_{l'})$ for some $l' \in \{1, \dots, k-1\}$.

Since $(u,v_k)$ is bi-directional, $\edge{u,v_k} \in \oY4$ and, by
Definition~\ref{de:anchors}, $\cpath{u}{v_{l'},v_k}$ is a maximal
uni-directional canonical path ending at $v_k$. Since $(u,v_k)$ is not dual,
Step~\ref{st:spanner1} ensures that $\cpath{u}{v_{l'},v_k=v}$ is a path in
$H_8$.
Therefore, if $(u,v_{l'}) \in H_8$ then the $1$-standard (and thus $2$-standard)
path from $u$ to $v$ is in $H_8$.


If $(u,v_{l'}) \not\in H_8$ then $(u,v_{l'})$ must be a weak anchor and thus a dual
edge. Since $(u,v_{l'})$ lies clockwise from $(u,v_k)$ within cone $0$ of $u$,
$(u,v_{l'})$ must be the first edge in cone 0 of $u$, i.e. $l'=1$. Any other
anchor incident to $u$ in its cone $0$ would have to be a weak anchor and thus
a dual edge (Observation~\ref{ob:anchors1}{\em-(b)}). Since the first edge
in the cone is $(u,v_1)$ and the last is $(u,v_k)$, there cannot
be another anchor incident to $u$ in its cone $0$. So $(u,v_1)$ is a
start-of-odd-chain anchor chosen by $u$. By
Observation~\ref{ob:anchors2}{\em-(c)}, $\anchor_2(v_1) \in H_8$ and
$\anchor_2(v_1) \not= (v_1,u)$.
Since at least two edges of $\Y4$ are incident
to $v_1$ in its cone $2$, node $v_1$ has a fan
$(v_1,u_1), \dots, (v_1,u_{k'})$ with $k' \geq 2$ in its cone $2$. Let
$\anchor_2(v_1) = (v_1, u_{l'})$ for some $l' \in \{1,\dots,k'\}$.
Since, by Observation~\ref{ob:initial1.5}{\em-(a)}, $(v_1,u)$ is the
first edge in cone $2$ of $v_1$, $u=u_1$.
Because $\edge{v_1,u} \in \oY4$ (a consequence of the assumption that
$\edge{u,v_k}=\edge{u,v} \in \oY4$), by Definition~\ref{de:anchors}
$\cpath{v_1}{u_{l'},u_1=u}$ is a maximal uni-directional canonical path
ending at $u=u_1$.
Step~\ref{st:spanner1} ensures that all the edges on this path are in $H_8$
(in particular $(u_2,u_1=u)$ is in because $(u,v_1)$ is a start-of-odd-chain
anchor chosen by $u$). Therefore, the $2$-standard path from $u$ to $v=v_r$ is
in $H_8$. The bound on the length of the path follows from
Lemma~\ref{le:L1path}.
\end{proof}

\begin{lemma}
\label{le:unidirectional}
Let $(u,v)$ be a uni-directional canonical edge of $\Y4$ in cone $i$ of
$u$ such that $\edge{v,u} \in \oY4$ and $\anchor_i(u)$ is defined.
The $4$-standard pre-path in $H_8$ from $u$ to $v$ is a path in $H_8$
and has length at most $(3+\sqrt 2)^4 \cdot d_2(u,v)$.
\end{lemma}

\begin{proof}
If $\anchor_i(u) = (u,v)$ then $(u,v)$ must be
a selected anchor (only anchors that are dual may not be selected) and thus in
$H_8$ so the theorem holds trivially. So we can assume that
$\anchor_i(u) \not= (u,v)$ and that, w.l.o.g., it lies clockwise from $(u,v)$
within cone $i$ of $u$.

If $\anchor_i(u) = (u,v') \in H_8$ then the $2$-standard path from $u$ to
$v$ consists of edge $(u,v')$ and
$\cpath{u}{v',v}$, say $v'=v_{l'},v_{l'+1}, \dots, v_k=v$.
Step~\ref{st:spanner1} ensures that all the uni-directional canonical
edges on this path are in $H_8$ (in particular, edge $(v_{k-1},v_k=v)$ is
in because $(u,v_k=v)$ is a canonical edge and thus not dual).
By Lemma~\ref{le:bidirectional}, for every bi-directional edge in
$\cpath{u}{v',v}$ not in $H_8$, a $2$-standard path between its endpoints is
in $H_8$.
Therefore the $4$-standard pre-path in $H_8$ from $u$ to $v$ is a path from
$u$ to $v$.

If $\anchor_i(u) = (u,v') \not\in H_8$ then the $2$-standard path from
$u$ to $v$ consists of $\cpath{u}{v',v}$, say $v'=v_{l'},v_{l'+1}, \dots, v_k=v$,
$\anchor_{i+2}(v') = (v',u')$, and $\cpath{v'}{u',u}$, say
$u=u_1,u_2, \dots, u_{l'}=u'$. Furthermore, just as in the proof of
Lemma~\ref{le:bidirectional}, $(u,v'=v_{l'})$ must be a start-of-odd chain
anchor chosen by $u$.
Step~\ref{st:spanner1} ensures that all the uni-directional canonical
edges on $\cpath{u}{v',v}$ and $\cpath{v'}{u',u}$ are in $H_8$; in particular,
edge $(v_{k-1},v_k=v)$ is in because $(u,v_k=v)$ is a canonical edge and
thus not dual and $(u_2,u_1=u)$ is in because $\anchor_i(u)$ is a
start-of-odd-chain anchor. By Lemma~\ref{le:bidirectional},
for every bi-directional edge in $\cpath{u}{v',v}$ or
$\cpath{v'}{u,u'}$ that is not in $H_8$, a $2$-standard path between its
endpoints is in $H_8$.
Therefore the $4$-standard pre-path in $H_8$ from $u$ to $v$ is a path; the
bound on its length follows from Lemma~\ref{le:L1path}.
\end{proof}

\begin{proofof} {\bf Theorem~\ref{th:span}.}
The theorem holds trivially if $(u,v)$ is a selected anchor, so we
assume otherwise. W.l.o.g. we assume that $\anchor_i(u)$ is defined and that
$\anchor_i(u)$ is either clockwise from $(u,v)$ within cone $i$ of $u$ or that
$\anchor_i(u) = (u,v)$.

If $\anchor_i(u) = (u,v')$ is in $H_8$ then the $2$-standard path from $u$ to
$v$ consists of edge $(u,v')$ and
$\cpath{u}{v',v}$, say $v'=v_{l'},v_{l'+1}, \dots, v_r=v$.
Step~\ref{st:spanner1} ensures that all uni-directional canonical
edges on this path are in $H_8$ except for possibly edge $(v_{r-1},v_r=v)$;
if missing, this edge is a uni-directional canonical edge such that
$\edge{v_{r-1},v_r} \in \oY4$. By Lemma~\ref{le:unidirectional}, the
$4$-standard pre-path in $H_8$ from $v_r$ to $v_{r-1}$ is a path from $u$
to $v$. For every
bi-directional edge in $\cpath{u}{v',v}$ not in $H_8$, a $2$-standard path
between its end points is in $H_8$. Therefore the $6$-standard pre-path in $H_8$
from $u$ to $v$ is a path in $H_8$.

If $\anchor_i(u) = (u,v')$ is not in $H_8$ then the $2$-standard path from
$u$ to $v$ consists of $\cpath{u}{v',v}$, say $v'=v_1,v_2, \dots, v_r=v$,
$\anchor_{i+2}(v') = (v',u')$, and $\cpath{v'}{u',u}$, say
$u=u_1,u_2, \dots, u_{l'}=u'$. Step~\ref{st:spanner1} ensures that all
uni-directional canonical
edges on $\cpath{u}{v',v}$ and $\cpath{v'}{u',u}$ are in $H_8$ except for
possibly $(v_{r-1},v_r=v)$ and $(u_2,u_1=u)$.
By Lemma~\ref{le:unidirectional}, the $4$-standard pre-paths in $H_8$ from $v_r$
to $v_{r-1}$ and from $u_1$ to $u_2$ are paths in $H_8$. By
Lemma~\ref{le:bidirectional},
for every bi-directional edge in $\cpath{u}{v',v}$ or
$\cpath{v'}{u,u'}$ that is not in $H_8$, a $2$-standard path between the
endpoints is in $H_8$. Therefore, the $6$-standard pre-path in $H_8$ from
$u$ to $v$ is a path and the bound on its length follows from
Lemma~\ref{le:L1path}.
\end{proofof}

\section{Reducing the maximum degree bound to 4}
\label{se:H4}

In order to reduce the degree of $H_8$, we need to remove for every cone
with a charge of 2 at least one edge of $H_8$ that contributes to the charge.
Lemma~\ref{le:charge2} describes the three cases in which a cone receives a
charge of 2. We name the pair of non-anchor uni-directional canonical edges of
$u$ in its cone $i$ satisfying the condition in case {\em (a)} of
Lemma~\ref{le:charge2} an {\em edge pair in cone $i$ of $u$}. We also name
the non-anchor uni-directional canonical edge of $u$ satisfying case {\em (b)}
the {\em duplicate first edge in cone $i$ of $u$} and the non-anchor
uni-directional canonical edge of $u$ satisfying case {\em (c)}
the {\em duplicate last edge in cone $i$ of $u$}.

We have shown in Theorem~\ref{th:span} that there is a
$6$-standard, or simply standard, path in $H_8$ between the endpoints of every
edge in $\Y4$. These paths together satisfy the following:
\begin{observation}
\label{ob:standard}
For every node $u$ with a fan $(u,v_1), \dots, (u,v_k)$, with $k \geq 2$,
in its cone $i$ and $\anchor_i(u) = (u,v_{l'})$ for some
$l' \in \{1,\dots,k\}$, the following hold for the set of all
$6$-standard paths in $H_8$:
\begin{enumerate}[(a)]

\item For every edge pair $(v_{r-1},v_{r}), (v_r,v_{r+1})$ of $u$
(where $r \in \{2,\dots,k-1\}$) if a standard path in $H_8$ contains both
edges then they must appear consecutively in the path, and if a standard path
contains just one of them then the standard path must be from $u$ to $v_r$.

\item If $(v_2,v_1)$ is a duplicate first edge (of $u$) then
no standard path in $H_8$ can contain $(v_2,v_1)$ other than the standard path
from $u$ to $v_1$.

\item If $(v_{k-1},v_k)$ is a duplicate last edge (of $u$) then
no standard path in $H_8$ can contain $(v_{k-1},v_k)$ other than the standard
path from $u$ to $v_k$.

\end{enumerate}
\end{observation}

\begin{proof}
By Definition~\ref{de:standard},
if $(v_r,v_{r+1})$, for some $r \in \{1,\dots,k-1\}$, is a non-anchor
uni-directional canonical edge that appears on a $2d$-standard path in $H_8$,
say path $p$, from
some node $u'$ to another node $v'$, then either 1) $(v_r,v_{r+1})$ belongs to
the $2$-standard path $p_1$ from $u'$ to $v'$ or 2) $(v_r,v_{r+1})$ belongs to the
$2(d-1)$-standard path in $H_8$ from $w$ to $w'$ for some canonical edge
$(w,w')$ in $p_1$ that is not in $H_8$.
In case 1), by definition of $2$-standard paths, $u'$ must be $u$ and the
subpath of $p_1$ starting at
$u$ and ending with edge $(v_r,v_{r+1})$ is the $2$-standard path from $u$ to
$v_r$, if $l' \geq r+1$, or to $v_{r+1}$, if $l' \leq r$. By
Definition~\ref{de:standard}, the $2d$-standard path from $u$ to $v_r$,
if $l' \geq r+1$, or $v_{r+1}$, if $l' \leq r$, is contained in $p$.
In case 2), we apply recursion until we obtain that $(v_r,v_{r+1})$
belongs to the $2$-standard path from $w$ to $w'$ for some canonical
edge $(w,w')$ not in $H_8$. Using the above argument, $w$ must be $u$,
$(v_r,v_{r+1})$ must be contained in $\cpath{u}{v_{l'},w'}$, and $(u,w')$
must be the first or last edge in cone $i$ of $u$.

If $(v_2,v_1)$ is a duplicate first edge of $u$ then $(v_1,u)$ is not an anchor
and $l' > 1$. Hence the only $1$-standard path that uses $(v_2,v_1)$ is the
$1$-standard path from $u$ to $v_1$ and the only
$2$-standard path that uses $(v_2,v_1)$ is the $2$-standard path from $u$
to $v_1$. Since $(u, v_1)$ is in $H_8$, $(v_2,v_1)$ appears only in one standard
path in $H_8$, the one from $u$ to $v_1$. This proves part {\em (b)} and, by
symmetry, {\em (c)}.
To prove part {\em (a)}, suppose $(v_{r-1},v_{r}), (v_r,v_{r+1})$,
for some $r \in \{2,\dots,k-1\}$, is an edge pair of $u$. By
Observation~\ref{ob:anchors1}{\em-(a)}, $r \not= l'$. If only
$(v_{r-1},v_r)$ appears in standard path $p$ then $l' < r$ and, since
$(u,v_r)$ is not canonical, $p$ must be the standard path
from $u$ to $v_r$. Similarly, if only $(v_{r+1},v_r)$ appears in standard path
$p$ then $l > r$ and $p$ must also be the standard path from $u$ to $v_r$.
Finally, if both edges appear in standard path $p$ then they must be, by
Definition~\ref{de:standard}, consecutive edges in the path.
\end{proof}

The following observation applies to $H_8$ but we find it useful to state it
more generally:
\begin{observation}
\label{ob:standard2}
Let $H$ be a subgraph of $\Y4$ that includes all selected edges, let
$(u,v_1), \dots, (u,v_k)$, with $k \geq 2$, be the fan of cone $u$ in its
cone $i$, and let $\anchor_i(u) = (u,v_{l'})$ for some $l' \in \{1,\dots,k\}$.
If $(u,v)$, where $v = v_1$ or $v = v_k$, is a non-anchor
uni-directional canonical edge then $(u,v)$ cannot appear on the
$6$-standard pre-path in $H$ from $u$ to $v$.
\end{observation}

\begin{proof}
We assume w.l.o.g. that $i=0$ and $(u,v) = (u,v_k)$. In
that case, $(u,v_k)$ is a canonical edge of some node $s$ in its cone 3 where
$s$ lies within cone $1$ of $u$. Note that edge $(u,v_k)$ and anchor edge
$(u,v_{l'})$ lie in cone 0 of $u$ and cones 2 of $v_k$ and $v_{l'}$,
respectively.

By applying Observation~\ref{ob:initial}{\em-(d)} to cone 0 of
$u$ and, if $(u,v_{l'})$ is not selected and thus $l' = 1$, cone 2 of $v_1$,
all non-anchor canonical edges on the $2$-standard path $p$ from $u$ to $v_k$
lie in cones 1 or 3 of their endpoints. By repeating this for every non-anchor
canonical edge on $p$, we deduce that all non-anchor canonical edges on the
$4$-standard pre-path in $H$ from $u$ to $v_k$ but not on the $2$-standard
pre-path in $H$ lie in cones 0 and 2 of their endpoints. We continue one more
time to find that all non-anchor canonical edges on the $6$-standard pre-path
in $H$ from $u$ to $v_k$ but not on the $4$-standard pre-path in $H$ lie in
cones 1 and 3 of their endpoints. This implies that if $(u,v_k)$ does not
appear in the $4$-standard pre-path in $H$ from $u$ to $v_k$, it will not
appear in the $6$-standard pre-path in $H$ either. Therefore we only need to
show that $(u,v_k)$ does not appear in the $4$-standard pre-path in $H$ from
$u$ to $v_k$.

The non-anchor canonical edges on the $2$-standard path from $u$ to $v_k$ are
1) canonical edges of $u$ in its cone 0 and, if $(u,v_{l'})$ is not selected
and thus $l' = 1$, 2) canonical edges of node $v_1$ in its cone 2 and that lie
in cone 3 of $u$. Consider a non-anchor canonical edge $(v_r,v_{r+1})$ of $u$
in its cone 0. Either $\anchor_1(v_r)$ is defined and
$\anchor_1(v_r) = (v_r, w')$
or $\anchor_3(v_{r+1})$ is defined and $\anchor_3(v_{r+1}) = (v_{r+1}, w')$.
Note that $w'$ must lie in cone 0 of $u$. The non-anchor canonical edges on
the $2$-standard path between $v_r$ and $v_{r+1}$ are all
canonical edges of nodes $v_r$ and $w'$ if $\anchor_1(v_r)$ is defined or
$v_{r+1}$ and $w'$ if $\anchor_3(v_{r+1})$ is defined. Since $s$ lies in cone
$1$ of $u$, $v_r$, $v_{r+1}$, and $w'$ must be different from $s$.
If we now consider a canonical edge $(w_1,w_2)$ of $v_1$ in
its cone 2, we can similarly show that non-anchor canonical
edges on the $2$-standard path between $w_1$ and $w_2$ are canonical edges
of nodes lying in cone 3 of $u$ and thus cannot be $s$.

This implies that $(u,v_k)$, a canonical edge of $s$, cannot appear on
a $4$-standard pre-path in $H$ from $u$ to $v_k$.
\end{proof}

By Observation~\ref{ob:standard}{\em-(b)}, if $(v_2,v_1)$ is a duplicate first
edge in cone $i$ of $u$ then no standard path in $H_8$ other than the
$4$-standard path (by Lemma~\ref{le:unidirectional}) from $u$ to $v_1$ uses
edge $(v_2,v_1)$. Furthermore, edge $(v_1,u) \in H_8$ by definition (of
duplicate first edge). Therefore, {\em as long as we keep edge
$(v_1,u)$}, we can remove $(v_2,v_1)$ from $H_8$ without breaking any standard
path other than the one from $u$ to $v_1$ and without increasing the stretch
factor bound from Theorem~\ref{th:span}. By symmetry, a similar insight can be
made about Observation~\ref{ob:standard}{\em-(c)} and duplicate last edge
$(v_{k-1},v_k)$. To generalize the discussion that follows, we will
call an edge a {\em duplicate edge of $u$} if it is a duplicate
first or last edge in some cone of $u$.

If $(v_2,v_1)$ is a duplicate first edge of $u$ in its cone $i$ then,
by Observation~\ref{ob:initial1.5}{\em-(b)}, $(v_2,v_1)$ is the last
edge in cone $i+1$ of $v_1$ and $(u,v_1)$ is a non-anchor uni-directional
canonical edge
in $H_8$ with $\edge{v_1,u} \in \oY4$. Then, either $(v_1,u)$ is a duplicate
last edge in cone $i+1$ of some node $w$
lying in cones $i-1$ of $u$ and of $v_1$ or it is not a duplicate
edge at all. This insight, along with the symmetric one regarding
$(v_{k-1},v_k)$ and $(v_k,u)$, motivates this definition:

\begin{definition} \rm
A {\em chain of duplicate edges} in $H_8$ is a path
$w_1, \dots , w_k, w_{k+1}$ of maximal length in $H_8$ in which
every edge $(w_{l-1},w_l)$ is a duplicate edge of $w_{l+1}$ for
every $l=2,3,\dots,k$ and $(w_k,w_{k+1})$ is not a duplicate
edge. Edge $(w_k,w_{k+1})$ is referred to as the {\em end edge} of the chain
and $k$ is the length of the chain.
\end{definition}
For instance in Figure~\ref{fi:bigexample}-(c), edge $(u_3,u_2)$ is the end
edge of the chain of duplicate edges $u_3,u_2,u_{12},u_{10}$.
\begin{observation}
\label{ob:dchain}
The following hold for chains of duplicate edges:
\begin{enumerate}[(a)]
\item A chain of duplicate edges does not form a cycle.
\item Every duplicate edge belongs to exactly one chain of duplicate edges.
\item Every non-anchor uni-directional canonical edge in $H_8$ that is not a
duplicate edge is the end edge of exactly one chain, possibly of length 1.
\end{enumerate}
\end{observation}

\begin{proof}
Let $(w_{l-1},w_l)$ be a duplicate (w.l.o.g., first) edge of
$w_{l+1}$ in its cone
$i$ and let $w_{l-2}$ be the predecessor of $w_{l-1}$ in a chain of duplicate
edges containing $(w_{l-1},w_l)$. Note that $w_{l-2}$ lies in cone $i$ of
$w_{l-1}$. On the other hand, $w_l$, $w_{l+1}$, and all other successors of
$w_{l-1}$ on the chain lie in cone $i+2$ or $i+3$ of $w_{l-1}$. This proves
part {\em (a)}. Because $(w_{l-2},w_{l-1})$ and $(w_l,w_{l+1})$ must be first
edges in cones $i$ of $w_{l-1}$ and $w_{l+1}$, respectively, the edges that
is before or after $(w_{l-1},w_l)$ in a chain of duplicate edges are uniquely
defined proving part {\em (b)}. For the same reason, there can be only
one chain whose end edge is a given uni-directional non-duplicate
edge and part {\em (c)} follows.
\end{proof}

By Observation~\ref{ob:dchain}, we can partition all non-anchor uni-directional
canonical edges of $H_8$ into chains of duplicate edges. To construct our final
spanner we first remove every other edge in every chain as follows:

\begin{step}
\label{st:notlast}
For every chain of duplicate edges $w_1,w_2, \dots, w_{k+1}$ we
remove from $H_8$ every other edge in the chain starting with $(w_{k-1},w_k)$,
i.e. $(w_{k-1},w_k)$, $(w_{k-3},w_{k-2})$, ...
\end{step}
We further remove edge pairs $(v_{r-1},v_r), (v_{r+1},v_r)$
and replace them with a shortcut:

\begin{step}
\label{st:last}
For every node $u$, every cone $i$ of $u$, and every edge pair
$(v_{r-1},v_r), (v_{r+1},v_r)$ in cone $i$ of $u$, we remove $(v_{r-1}, v_r)$
and $(v_r, v_{r+1})$ from $H_8$, we add a new (straight-line) edge between $v_{r-1}$ and
$v_{r+1}$, and charge edge $(v_{r-1}, v_{r+1})$ to the cones of $v_{r-1}$ and
$v_{r+1}$ in which the edge lies. We
call this edge a {\em shortcut} between $v_{r-1}$ and $v_{r+1}$.
We also call $v_r$ a {\em cut-off} node with respect to $u$.
\end{step}
Let $H_4$ be the resulting graph. In Figure~\ref{fi:bigexample}-(d), edges
$(u_8,u_{14})$ and $(u_1,u_{14})$ have been removed during step 4. All
other edges present in $H_8$ but not in $H_4$ have been removed during
step 3. Moreover, the shortcut edge $(u_8,u_1)$ has been added during
step 4.

Before we prove the main theorem, we show that the following strenghtening of
Theorem~\ref{th:span} holds for subgraph $H_6$ of $H_8$ that is obtained
after applying Step~\ref{st:notlast}, but not Step~\ref{st:last}, to $H_8$:
\begin{lemma}
\label{le:H6}
Let $(u,v)$ be an edge of $\Y4$ lying in cone i of u. There is a $6$-standard
path in $H_6$ from $u$ to $v$ or from $v$ to $u$ of length no more than
$(3+\sqrt 2)^6 \cdot d_2(u,v)$.
\end{lemma}

\begin{proof}
Every standard path in $H_8$ that contains no edge removed in
Step~\ref{st:notlast} is in $H_6$, so we only need to consider
standard paths that do contain a removed edge, i.e. get broken.
Let $(v_2,v_1)$ be a duplicate first edge, of some node $u$, that is in $H_8$
but not in $H_6$ (the argument for a duplicate last edge is symmetric).  By
Observation~\ref{ob:standard}{\em-(b)}, the only standard path in $H_8$
that gets broken by the removal of $(v_2,v_1)$ is the standard path
in $H_8$ from $u$ to $v_1$. Since $(u,v_1)$ is a uni-directional canonical
edge in $\Y4$, by Lemma~\ref{le:unidirectional} this standard path must
be a $4$-standard path $p$ in $H_8$ from $u$ to $v_1$. Note that the
standard path in $H_8$ from $v_1$ to $v_2$, which happens to be a $4$-standard
path as well by Lemma~\ref{le:unidirectional}, does not get broken
(Observation~\ref{ob:standard2}) and
so it is also the $4$-standard path in $H_6$ from $v_1$ to $v_2$.
Since $(u,v_1)$ is a uni-directional canonical edge and since
$(v_1,v_2)$ belongs to the $2$-standard path from $u$ to $v_1$, the
path from $u$ to $v_1$ consisting of the $4$-standard path from $u$ to $v_1$ in
$H_8$ with edge $(v_2,v_1)$ being replaced by the $4$-standard path from $v_1$
to $v_2$ in $H_8$ is a $6$-standard path in $H_6$ from $u$ to $v_1$.
\end{proof}

\begin{theorem}
\label{th:final}
$H_4$ is plane spanner of the Euclidean graph of maximum degree at most 4
and stretch factor at most $\sqrt{4+2\sqrt{2}} (1+\sqrt{2})^2 (3+\sqrt 2)^6$.
\end{theorem}

\begin{proof}
We first argue planarity, which could potentially be affected
by shortcut edges since they are the only edges in $H_4$ not in $\Y4$.
Suppose $(v_{r-1},v_{r+1})$ is a shortcut in $H_4$ that was put in because
edge pair $(v_{r-1},v_r),(v_{r+1},v_r)$ in cone $i$ of $u$ was removed from
$H_8$ and $v_r$ is a cut-off node of $u$. Note that this implies that
$\edge{v_{r-1},v_r},\edge{v_{r+1},v_r} \in \oY4$. The only edge of $\Y4$ that
$(v_{r-1},v_{r+1})$
intersects is $(v_r,u)$. That edge is a middle edge and hence not canonical.
Therefore the only way for $(u,v_r)$ to be in $H_4$ is if it was the anchor
chosen by $u$.
By Observation~\ref{ob:anchors1}{\em-(a)}, that would contradict the
orientation of $(v_{r-1},v_r)$ and $(v_{r+1},v_r)$ in $\oY4$.
Furthermore, because a shortcut is added only between nodes $v_{r-1}$
and $v_{r+1}$ such that $(v_{r-1},v_r)$ and $(v_{r+1},v_r)$ are uni-directional,
$\edge{v_{r-1},v_r} \in \oY4$, and $\edge{v_{r+1},v_r} \in \oY4$, it is not
possible for two shortcut edges to intersect.

To argue the degree bound we only need to consider the three cases of
Lemma~\ref{le:charge2}. In case {\em (b)}, either $(v_1,u)$ or $(v_2,v_1)$
is no longer in $H_4$ and the charge in cone $i+2$ is reduced from 2 to
1. The same is true for case {\em (c)}. In case {\em (a)}, the charge in cone
$i+2$ of $v_r$ is reduced from 2 to 0. The added shortcut edge
$(v_{r-1},v_{r+1})$
replaces the removed edges $(v_{r-1},v_r)$ and $(v_{r+1},v_r)$ and does not
change the charge in the cones of $v_{r-1}$ and $v_{r+1}$ containing the shortcut
edge.

In order to prove the stretch factor bound, by Lemma~\ref{le:H6} we only need
to consider the standard paths in $H_6$ that are broken by the removal of
edges in Step~\ref{st:last}.
Consider an edge pair $(v_{r-1},v_r), (v_{r+1},v_r)$ of some node $u$ that
is in $H_8$ but not in $H_4$. Both edges are uni-directional canonical edges
and thus, by Observation~\ref{ob:initial1.5}{\em-(b)}, they are canonical
edges of just one node ($u$).
Because $\edge{v_{r-1},v_r}, \edge{v_{r+1},v_r} \in \oY4$, neither
edge can be a duplicate edge and therefore both are in $H_6$ as well.
If a standard path in $H_8$ contains both edges,
by Observation~\ref{ob:standard}{\em-(a)} the edges must be consecutive
in the standard path and therefore shortcut edge $(v_{r-1},v_{r+1}) \in H_4$
may be used instead of the missing edge pair which actually shortens the path
and so these paths are not broken. If a standard path in $H_8$
uses just one of $(v_{r-1},v_r)$ or $(v_{r+1},v_r)$ then by
Observation~\ref{ob:standard}{\em-(a)} the standard path must be the
one from $u$ to node $v_r$, a cut-off node with respect to $u$ in $H_4$.
Therefore, to complete the proof of the theorem, we only need to show that
a short path exists in $H_4$ between $u$ and cut-off node $v_r$ (with respect
to $u$). We assume w.l.o.g. that the standard path in $H_8$ from $u$ to
$v_r$ uses edge $(v_{r-1},v_r)$.

By Observation~\ref{ob:standard2}, edge $(v_{r-1},v_r)$ cannot appear
in the $6$-standard path in $H_6$ from $v_r$ to $v_{r-1}$ so its removal in
Step~\ref{st:last} does not break that path. Furthermore, since $v_r$ is a
cut-off node with respect to $u$ then $(u,v_r)$ must be a middle edge and
therefore $v_{r-1}$ cannot be a cut-off node of $v_r$ (since $(v_r,v_{r-1})$ is
canonical and thus not a middle edge). So, the $6$-standard path in $H_6$
between $v_r$ and $v_{r-1}$ must still exist, with shortcuts replacing any
edge pairs on the path, in $H_4$. By Lemma~\ref{le:L1path}{\em-(b)} and
Lemma~\ref{le:H6}, the length of this path is at most
$\sqrt{2}(3+\sqrt 2)^6 \cdot d_2(u,v)$. Since $v_{r-1}$ cannot be a cut-off
node of $u$ (because
$\edge{v_r,v_{r-1}} \not\in \oY4$), the $6$-standard path in $H_6$ from
$u$ to $v_{r-1}$ is also in $H_4$. The length of this path is no more than
$(3+\sqrt 2)^6 \cdot d_2(u,v)$. Therefore there is a path in $H_4$ from $u$
to $v_r$ of length bounded by $(1+\sqrt{2})(3+\sqrt 2)^6 \cdot d_2(u,v_r)$
and so $H_4$ is a spanner with stretch factor at most $\sqrt{4+2\sqrt{2}} (1+\sqrt{2})^2 (3+\sqrt 2)^6$ by
Lemma~\ref{ob:initial2}.
\end{proof}



\section{Conclusion}

The question that this paper addressed is:
{\em What is the smallest maximum degree that can be achieved for plane
spanners of complete Euclidean graphs?} The main result of this paper allows
this question to be reformulated as follows: {\em Is it always possible to
construct a maximum degree 3 plane spanner of complete Euclidean graphs?}

Given the $L_\infty$-Delaunay triangulation of the considered point set,
the construction of $H_4$ can be done in linear time. The stretch factor
bound from Theorem~\ref{th:final} is a very rough bound on the stretch factor
of $H_4$. Our main goal was to present a simpler proof showing that $H_4$ is
spanner of maximum degree four. The bound can be easily improved
with a more careful analysis (leading to a proof with more cases to consider). We have written a program that constructs spanner $H_4$ on a set of points
(see \texttt{http://www.labri.fr/$\sim$bonichon/deg4}) and we have failed to
to obtain examples that give a spanner with stretch factor greater than 10.
We believe that the real spanning ratio is much lower and
that this construction may have practical applications.

There exists a distributed algorithm that can compute a plane spanner
of maximum degree 6 (and we denote this spanner $H'_6$) with a constant
number of rounds~\cite{BGHP10}. For the construction of $H_4$, the number of
necessary rounds is bounded (from below) by the length of longest weak anchor
chain and the length of the longest duplicated chain, and such a chain can
have a linear number of vertices.

In \cite{bose2012competitive2} it is shown that there exists a routing
algorithm on $H'_6$ with a bound stretch factor. We leave open the
question whether or not it is possible to obtain a similar result on $H_4$.
The construction of $H'_6$ has been extended to constraints graphs
in~\cite{bose2012plane}. Once again, we leave open the question that
whether or not it is possible to obtain similar results on $H_4$.


\section{Figures}

\newcommand{\scaleExample}{0.018}

\begin{figure}[!b]
\noindent
\begin{minipage}{.5\textwidth}
\centering
\begin{tikzpicture}[scale=\scaleExample]
  \foreach \pos/\name/\hh/\hhh/\hhhh/\hhhhh/\ff/\fff/\ffff/\fffff/\side  in
  \nodeList
 \node[vertex, label=\side:$u_{\name}$] (v\name) at \pos {};

 \foreach \src/\dst/\color in \YaoEdges
 \draw[yaoEdge, color=\color] (v\src) -- (v\dst) ;

\draw[edge, blue, dashed] (v4)--(v23);
\draw[edge, red!70, dashed] (v26)--(v14);
\draw[edge, red!70, dashed] (v22)--(v7);
\end{tikzpicture}

\centering $(a)$
\end{minipage}%
\begin{minipage}{.5\textwidth}
\centering

\begin{tikzpicture}[scale=\scaleExample]
  \foreach \pos/\name/\hh/\hhh/\hhhh/\hhhhh/\ff/\fff/\ffff/\fffff/\side  in
  \nodeList
 \node[vertex, label=\side:$u_{\name}$] (v\name) at \pos {};


 \foreach \src/\dst/\color in \anchorEdges
 \draw[anchorEdge, dashed, color=\color] (v\src) -- (v\dst) ;

 \foreach \src/\dst/\color in \StrongEdges
 \draw[strongEdge, color=\color] (v\src) -- (v\dst);

\end{tikzpicture}

\centering $(b)$
\end{minipage}

\noindent
\begin{minipage}{.5\textwidth}
\centering
\begin{tikzpicture}[scale=\scaleExample]
 \foreach \src/\dst/\color in \stepTwoEdges
 \draw[addedEdge, color=green] (v\src) -- (v\dst);

 \foreach \src/\dst/\color in \StrongEdges
 \draw[strongEdge, color=\color] (v\src) -- (v\dst);

 \foreach \src/\dst/\color in \pseudoStrongEdges
 \draw[pseudoStrongEdge, dashed, color=\color] (v\src) -- (v\dst);

 \foreach \pos/\name/\hh/\hhh/\hhhh/\hhhhh/\ff/\fff/\ffff/\fffff/\side  in
  \nodeList
  \node[vertex, label=\side:\scalebox{0.75}{$\charge{\hh}{\hhh}{\hhhh}{\hhhhh}$}] (v\name) at \pos {};

\end{tikzpicture}

\centering $(c)$
\end{minipage}%
\begin{minipage}{.5\textwidth}
\centering
\begin{tikzpicture}[scale=\scaleExample]
 \foreach \src/\dst/\color in \stepThreeEdges
 \draw[addedEdge, color=green] (v\src) -- (v\dst);

 \foreach \src/\dst/\color in \StrongEdges
 \draw[strongEdge, color=\color] (v\src) -- (v\dst);

 \foreach \src/\dst/\color in \pseudoStrongEdges
 \draw[pseudoStrongEdge, dashed, color=\color] (v\src) -- (v\dst);

 \foreach \src/\dst in \shortcuts
 \draw[shortcutEdge] (v\src) -- (v\dst);

  \foreach \pos/\name/\hh/\hhh/\hhhh/\hhhhh/\ff/\fff/\ffff/\fffff/\side  in
  \nodeList
 \node[vertex, label=\side:\scalebox{0.75}{$\charge{\ff}{\fff}{\ffff}{\fffff}$}] (v\name) at \pos {};
\end{tikzpicture}

\centering $(d)$
\end{minipage}

\caption{
$(a)$ The $L_\infty$-Delaunay triangulation of $P = \{u_0,u_1,\dots,u_{28}\}$.
  An edge lying in cone $0$ of one endpoint and cone $2$ of the other is
  colored red and an edge lying in cone $1$ of one endpoint and cone $3$ of the
  other is colored blue; the oriented edges belong to $\oY4$ while the dashed
  edges do not.
$(b)$ The anchor edges. An anchor $(u_i,u_j)$ is shown to be oriented from $u_i$
  to $u_j$ if it is chosen by $u_i$; solid edges are strong anchors and dashed
  edges are weak anchors.
$(c)$ Graph $H_8$. The label at a node shows the charge of each cone of that
  node; non-anchor uni-directional canonical edges are dotted green.
$(d)$ Graph $H_4$. The undirected black edge is a shortcut edge.}
\label{fi:bigexample}
\end{figure}
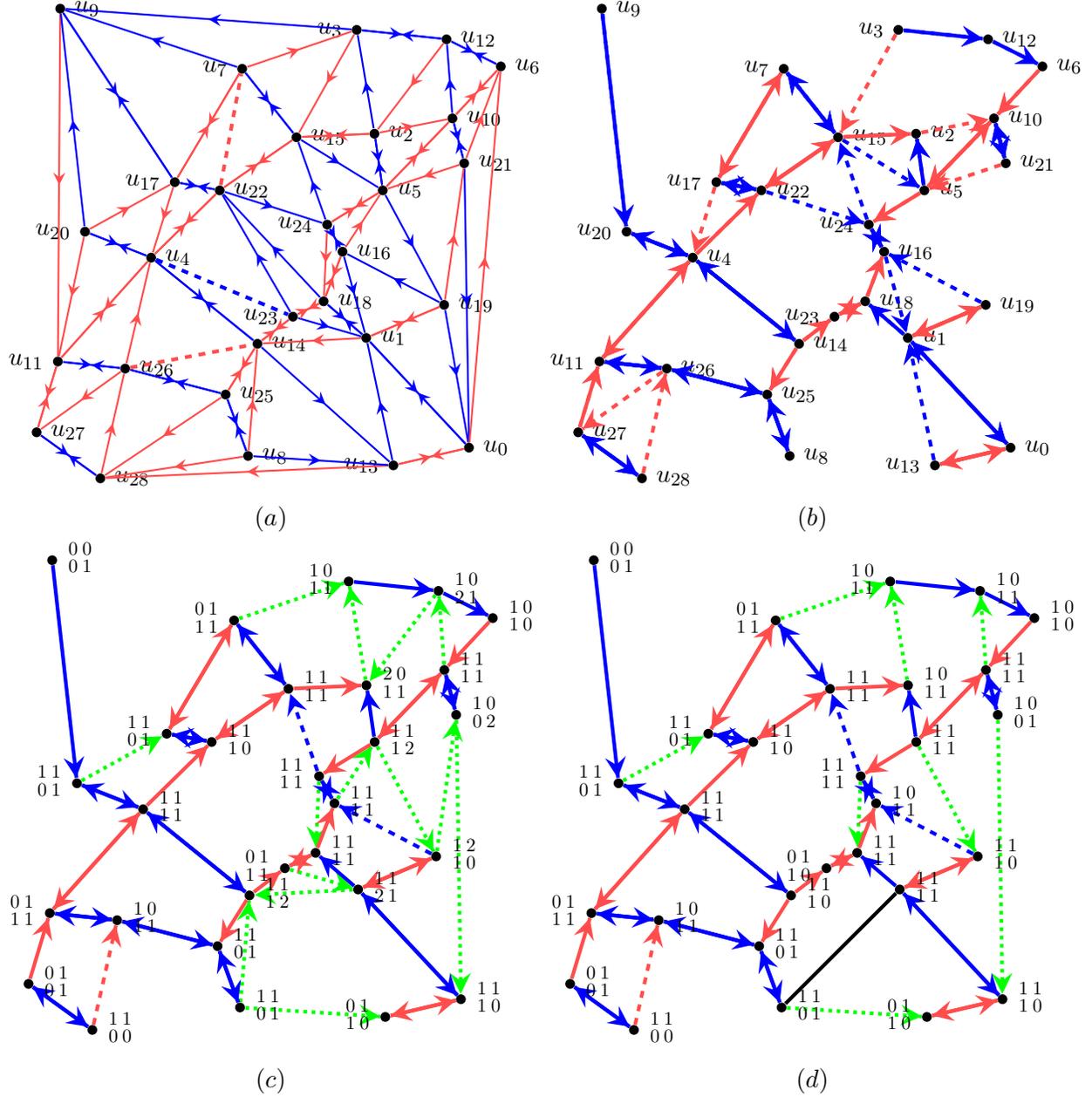

\newpage

\bibliographystyle{alpha}
\bibliography{bib}



\end{document}